\documentclass[journal, onecolumn]{IEEEtran}

\usepackage[noadjust]{cite}
\usepackage{amsthm}
\usepackage{amssymb}
\usepackage{bbm}

\pdfoutput=1

\ifCLASSINFOpdf
 \usepackage[pdftex]{graphicx}
 \graphicspath{{.}} %{generate_fig/result/simu_plot/}{generate_fig/result/dist_tree/}{generate_fig/result/info_graph/} {figs/}}
 \DeclareGraphicsExtensions{.pdf,.jpeg,.png}
 \else
 \usepackage[dvips]{graphicx}
 \graphicspath{{figs/}}
 \DeclareGraphicsExtensions{.eps,.pdf}
 \fi

\usepackage{amsmath}
\usepackage{algorithm}
\usepackage{algpseudocode}
\usepackage{array}
\ifCLASSOPTIONcompsoc
  \usepackage[caption=false,font=normalsize,labelfont=sf,textfont=sf]{subfig}
\else
  \usepackage[caption=false,font=footnotesize]{subfig}
\fi
\usepackage{url}
\usepackage[breaklinks]{hyperref}
\usepackage{color}
\usepackage{tensor}

\newtheorem{definition}{Definition}[section]

\newtheorem{lemma}{Lemma}
\newtheorem*{thm}{Theorem}

\newcommand{\salza}{\mbox{SALZA}}
\newcommand{\NSD}[1]{\ensuremath{\mbox{NSD}_{#1}}}

\newcommand{\alphabet}{\mathcal{A}}
\newcommand{\minlen}[2]{\ensuremath{\log^{#2}}_{|\alphabet_{#1}|}\left|#1\right|}

\newcommand{\allrefs}[1]{\ensuremath{\mathcal{L}_{#1}}}

\newcommand{\given}[4]{\ensuremath{{#1}_{\mbox{\small #2}}|^{\mbox{\tiny #3}}{#4}}}
\newcommand{\givenany}[2]{\ensuremath{{#1}\wr {#2}}}

\begin{document}
\title{\salza{}: Soft algorithmic complexity estimates for clustering and causality inference}

\author{Marion~Revolle,
        Fran\c cois~Cayre,
        and~Nicolas~Le~Bihan%
        \thanks{Authors email: \texttt{first.last@gipsa-lab.grenoble-inp.fr}.
          
          Marion Revolle is supported by a PhD grant from the French Ministry of Education and
          Research.

          Fran\c cois Cayre (corresponding author) is an assistant professor at Grenoble-INP.

          Nicolas Le Bihan is a research fellow at CNRS. 
        }% <-this % stops a space
}

\markboth{Submitted version}%
         {Revolle \MakeLowercase{\emph{et al.}}: \salza: Soft algorithmic
           complexity estimates for clustering and causality inference}



\maketitle

\begin{abstract}
  A complete set of practical estimators for the conditional, simple and joint algorihmic complexities
  is presented, from which a semi-metric is derived. Also, new directed information estimators are
  proposed that are applied to causality inference on Directed Acyclic Graphs. The performances of these
  estimators are investigated and shown to compare well with respect to the state-of-the-art Normalized
  Compression Distance (NCD \cite{li:similarity:metric}).
\end{abstract}

%\begin{IEEEkeywords}
%  Normalized information distance, causality inference
%\end{IEEEkeywords}

% \fi
\IEEEpeerreviewmaketitle

\section{Introduction}\label{sec:intro}

\IEEEPARstart{T}he use of compression algorithms to extract information from sequences or strings has
found applications in various fields \cite{li:kolmogorov} (and references therein), from biomedical and
EEG time series analysis \cite{zhang:eeg:kolmogorov}\cite{petrosian:eeg:kolmogorov} to languages
classification \cite{cilibrasi:clustering:compression} or species clustering
\cite{cilibrasi:clustering:compression}. Since the early papers of Lempel and Ziv
\cite{lempel:complexity}, the use of Kolmogorov complexity (the length of a shortest program able to
output the input string on a universal Turing computer \cite{li:kolmogorov}) in place of entropy for
non-probabilistic study of bytes strings, signals or digitized objects has led to the developement of the
Algorithmic Information Theory (AIT) \cite{li:kolmogorov}. Amongst the very interesting ideas and
concepts in AIT, the possiblity of defining distances between objects to measure their similitudes,
{\em i.e.} how much {\em information} they share, is one of the most used in practice. 

The approach proposed here was initially motivated by getting rid of most limitations induced be the
use of a practical compressor in computing an information distance \cite{cebrian:ncd:pitfalls}.
Reimplementing a coder from scratch allowed to give access to a much richer information than the raw
length of a compressed file (the only information available when using ``off-the-shelf'' compressors).
Doing so also allows to propose new estimates for simple, conditional and joint algorihmic complexities.
Such estimates can lead to the definition of a soft information semi-distance, as well as directed
information estimates.

The maximum information distance between two strings $x$ and $y$ is defined as
\cite{bennett:distance}: 
\begin{equation}\label{eq:information:distance}
  E_1(x,y) = \mbox{max}\{K(x|y),K(y|x)\},
\end{equation}
where $K(x|y)$ denotes the conditional Kolmogorov complexity of $x$ given $y$: the size of a shortest
program able to output $x$ knowing $y$. The Kolmogorov complexity $K$ is known not to be computable on
a universal Turing computer. Note that we will later extend the meaning of the conditional sign in
Sec.~\ref{ssec:conddef}. 

In order to compare objects of different sizes, the Normalized Information Distance (NID) has been
proposed \cite{li:similarity:metric}:
\begin{equation}
  \mbox{NID}(x,y)=\frac{\mbox{max}\{K(x|y), K(y|x)\}}{\mbox{max}\{K(x),K(y)\}},
\end{equation}
where $K(x)$ denotes the Kolmogorov complexity of $x$: the size of a shortest program able to output $x$ 
on a universal computer. Since $K$ is not computable, one usually has to resort to two main
approximations, leading to the Normalized Compression Distance (NCD) \cite{li:similarity:metric}, which
is a practical embodiment of the NID. 

Let $xy$ be the concatenation of $x$ and $y$. The first approximation concerns $K(x|y)$, which then
reads \cite{li:similarity:metric}: 
\begin{equation}\label{eq:concat-approx}
  K(x|y)\approx K(xy)-K(y).
\end{equation}

The second approximation consists in using the output of a real-world compressor to estimate $K$.
Let $C(x)$ be the compressed version of $x$ by a given compressor, the NCD then reads
\cite{li:similarity:metric}:
\begin{equation}\label{eq:ncd}
  \mbox{NCD}(x,y)=\frac{C(xy)-\mbox{min}\{C(x), C(y)\}}{\mbox{max}\{C(x),C(y)\}}. 
\end{equation}

It is remarkable that \emph{any} compressor can be used to estimate the NCD (actually, the results
are consistent even across \emph{families} of compressors). However, it is our goal to show that a
particular family of compressors, namely those based on the Lempel-Ziv algorithm \cite{ziv:universal},
are more amenable than others (namely, block-based) to mitigate the limitations induced by
Eq.~(\ref{eq:concat-approx}).

Using real-world compressors brings mainly two major limitations:
\begin{itemize}
\item The size of the block for block-based compressors \cite{cebrian:ncd:pitfalls};
\item For LZ77-based compressors:
  \begin{enumerate}
  \item The size of the sliding window. For example, the \texttt{DEFLATE} \cite{deutsch:rfc1951}
    compressor has a sliding window of 32KiB. Obviously, when the length of $y$ is greater than the
    size of the sliding window, then $C(xy)$ may not always use strings from $x$ to encode $y$. Note
    that increasing the size of the sliding window (\emph{e.g.}, \cite{pavlov:lzma}) does not guarantee
    that only strings from $x$ will be used to encode $y$;
  \item The size of the substrings that can be found in $y$ to encode $x$ is limited (258 bytes for
    \texttt{DEFLATE}).
  \end{enumerate}
\end{itemize}

The rest of this paper is organized as follows: After introducing four possible types of conditional
information in Sec~\ref{ssec:conddef}, we propose in Sec.~\ref{ssec:salza-cond} a normalized
conditional algorithmic complexity estimate. We also give the expressions for simple and joint
complexity estimates (Sec.~\ref{ssec:salza-joint}) and come up with a normalized
information semi-distance (Sec.~\ref{ssec:salza-semi-metric}). This is in contrast with previous
approaches were the normalization was performed afterwards. Then in Sec.~\ref{ssec:app:clustering} we
compare our semi-distance with the NCD. Finally, we define our algorithmic complexity estimates for
directed information \ref{ssec:dirinfo} and apply them to causality inference on directed acyclic graphs
in Sec.~\ref{ssec:app:causality}.

\section{\salza}\label{sec:salza}

The rationale for \salza{} is to come up with a practical implementation of algorithmic complexity
estimates that could work continuously over any sequence of symbols. This discards block-based
compression method like {\em e.g.}, \texttt{bzip2} which uses the Burrows-Wheeler block transform. 

As a practical implementation of \cite{ziv:universal} targeted towards data compression,
\texttt{DEFLATE} has to finely tune the intricacy 
between the string searching stage and the entropy coding stage. For example, \texttt{DEFLATE} will
generally use a socalled \emph{lazy match} string searching strategy, meaning it will not
necessarily find the longest strings, but stop the search at appropriate lengths. This has the combined
effect of speeding up the overall compression time, and also it produces a lengths distribution that
is more peaked (and hence more amenable to compression). However, it artificially creates references to
shorter strings that hence are not the most meaningful to explain the input data. 

Therefore, we explicitly depart from the pure data compression approach to compute our algorithmic
complexity estimates, but we keep Ziv-Merhav \cite{ziv:relative} and LZ77 \cite{ziv:universal} as the
core building blocks of our strategy. In the following, \emph{all} substrings are the longest
substrings. For the same reason, we do not use the entropy coding stage. 

\subsection{Conditional information}\label{ssec:conddef}

\salza{} builds on an implementation of \texttt{DEFLATE} \cite{deutsch:rfc1951} that has been further
extended with an unbounded buffer in order to implement the computation of a conditional Kolomogorov
complexity estimate. \salza{} will factorize a string $x$ using conditional information that may come
from different sets of strings. The string of possible references is denoted $\mathcal{R}$
and its alphabet is $\alphabet_{\mathcal{R}}$. We use the following to denote the four possible
reference strings, each time with respect to the current encoding position of the lookahead buffer:

\begin{enumerate}
\item \given{y}{}{}{x}: $\mathcal{R}$ is only the past of $x$:
  
  This models the usual LZ77 operating mode when $x=y$ (needed in Sec.~\ref{ssec:salza-joint}),
  and $\alphabet_\mathcal{R}=\alphabet_x$;
\item \given{y}{}{+}{x}: $\mathcal{R}$ is all of $x$:
  
  This models the usual Ziv-Merhav operating mode (needed in Sec.~\ref{ssec:salza-cond}),
  and $\alphabet_\mathcal{R}=\alphabet_x$;
\item \given{y}{-}{}{x}: $\mathcal{R}$ is the past of both $x$ and $y$:
  
  This will be needed later on in Sec.~\ref{ssec:dirinfo},
  and $\alphabet_\mathcal{R}=\alphabet_x\cup\alphabet_y$;
\item \given{y}{-}{+}{x}: $\mathcal{R}$ is the past of $y$ and all of $x$:
  
  This will be needed later on in Sec.~\ref{ssec:salza-joint},
  and $\alphabet_\mathcal{R}=\alphabet_x\cup\alphabet_y$;
\end{enumerate}

These types of conditional informations are depicted in Fig.~(\ref{fig:conddef}). When the conditional
information is left unspecified, we will use $\givenany{x}{y}$ to stand for either type of conditional
information. 

\begin{figure}[!ht]
  \centering
  \def\svgwidth{\columnwidth}
  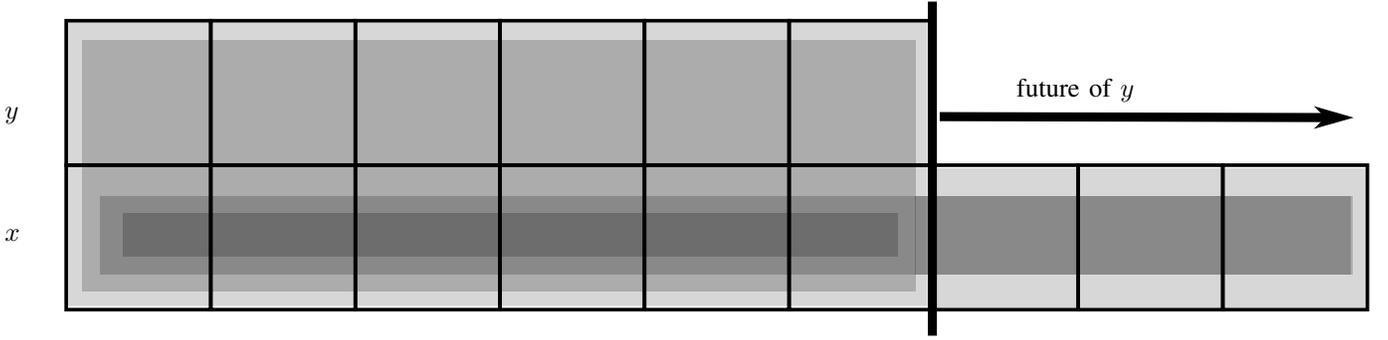
  \caption{Strings $\mathcal{R}$ for the conditional information (from which references are allowed), by
    darkening shades of gray: \given{y}{-}{+}{x}, \given{y}{-}{}{x}, \given{y}{}{+}{x} and
    \given{y}{}{}{x}. The thick vertical bar represents the position of the current lookahead buffer
    when encoding $y$.}
  \label{fig:conddef}
\end{figure}

Using any conditional information of the above, \salza{} will always produce \emph{symbols} of the form
$\left(l,v\right)$, which can be either:
\begin{itemize}
\item \textit{references}: $l>1$ is the length\footnote{\label{note:min3bytes}For implementation
  reasons, the actual minimum length of a reference is 3 bytes, like it is suggested in \texttt{DEFLATE}
  \cite{deutsch:rfc1951}. This is not mandatory but it follows from the data structure suggested for the
  dictionary (a hash table made of an array of linked lists). We believe this has a negligible impact on
  the overall results. Improving on this has been done in \cite{pavlov:lzma} and could be done here
  using (much) more involved and costly data structures like, \textit{e.g.}, Douglas Baskins' Judy
  arrays \cite{silverstein:judy}.} of a substring in the dictionary, and, although it is not used in this
  paper, $v$ is the offset in the dictionary at which bytes should start to be copied;
\item \textit{literals}: $l=1$ and $v$ is the literal in $x$ that should be copied to the output buffer.
\end{itemize}

It is important to notice that for any output symbol of \salza{}, one has that:
\begin{equation}\label{eq:lengths:all:geq1}
  \forall l, l \geq 1.
\end{equation}

Eventually, \salza{} will factorize a string $x$ in $n$ symbols by finding always the longest strings
(and $n$ will be our estimate of the \emph{relative complexity} \cite{ziv:relative} when \salza{} is
run in Ziv-Merhav mode):
$$x\rightsquigarrow\left(l_1,v_1\right)\ldots\left(l_n,v_n\right).$$
%$$ x \rightarrow \left( l_1 , v_1 \right) \ldots \left( l_n , v_n \right) $$

A specific feature of \salza{} is that it gives access to a much richer information than when using
an ``out of the box'' compressor. In particular, the lengths of the symbols produced during the factorization of a string allows to define a complexity measure that exhibits a
stronger discriminative power than those available in the state-of-the-art litterature. 

%%%%%%%%%%%%%%%%%%%%
%
% LUNDI 27/06 soir.
%
%%%%%%%%%%%%%%%%%%%%

\subsection{The \salza{} conditional complexity estimate}\label{ssec:salza-cond}

Let $\allrefs{\givenany{x}{y}}=\left\{l\right\}$ be the reference lengths produced by the
relative coding of $x$ given $y$ by \salza{} using any conditional information. Also, $|.|$ denotes
the length of a string or the cardinal of an alphabet. Note also that $x$ is defined over the alphabet
$\alphabet_x$ and $y$ is defined over $\alphabet_y$. When strings $x$ and $y$ are defined over the same
alphabet, it will be simply denoted $\alphabet$. 

When used in Ziv-Merhav mode ($\allrefs{\givenany{x}{y}}=\allrefs{\given{x}{}{+}{y}}$), \salza{}
will operate as follows:
\begin{itemize}
\item if $x$ and $y$ do not share the same alphabet, only literals will be produced;
\item if $x=y$, only one symbol will be produced, namely $\left(|x|,1\right)$.
\end{itemize}

\begin{definition} Set value.

  Let $f:\mathbb{N}^\star\rightarrow\mathbb{R}$ be a mapping and let $\mathcal{T}$ be a finite set of
  non-zero natural numbers. The image of $\mathcal{T}$ by $f$ is defined as:
  \begin{equation*}%\label{eq:s4t:value}
    |\mathcal{T}|_f = \sum_{s\in\mathcal{T}}f(s).
  \end{equation*}
\end{definition}

The notation $|\mathcal{T}|=|\mathcal{T}|_{\mathbbm{1}_\mathcal{T}}$ will also be used to denote the
cardinal of $\mathcal{T}$.

We now look for a normalized, conditional complexity estimate. Focusing on a conditional measure is the
first step to proposing the semi-distance later on in Sec.~\ref{ssec:salza-semi-metric} or constructing
directed informations estimates in \ref{ssec:dirinfo}. The very normalization, as in
\cite{li:similarity:metric}, is needed to compare objects of different sizes. In contrast with
\cite{li:similarity:metric}, our strategy is not to perform the normalization afterwards, but to use
directly Eq.~(\ref{eq:information:distance}) with an already normalized estimate. 

\begin{definition} Admissible function.
  
  A function $f:\mathbb{N}^\star\rightarrow\left[0,1\right]$ is said to be admissible iff it is
  monotonically increasing.
\end{definition}

In the following definition, the admissible function allows us to finely modulate the information that
is taken into account (see Sec.~\ref{ssec:salza:study}). 

\begin{definition}\salza{} conditional complexity estimate.
  
  Given an admissible function $f$, and two non-empty strings $x\in\alphabet_x$ and $y\in\alphabet_y$,
  the \salza{} conditional complexity estimate of $x$ given $y$,  denoted $S_f(\givenany{x}{y})$, is
  defined as:
  \begin{equation}\label{eq:salza-cond}
    S_f(\givenany{x}{y}) = \underbrace{\left(1-\frac{\sum_{\allrefs{\givenany{x}{y}}}l f(l)-
      (|\allrefs{\givenany{x}{y}}|_f-1)}{\left|x\right|}\right)}_{\mathcal{S}}
    \underbrace{\frac{|\allrefs{\givenany{x}{y}}|-1}{|x|}}_{\mathcal{Z}}.
  \end{equation}
  For simplicity, the dependency on $x$ $y$ and $f$ in the terms of the fatorization are omitted.
  Therefore, the form $S_f(\givenany{x}{y})={\mathcal S}{\mathcal Z}$ will be used in the sequel. 
\end{definition}

In Eq.~(\ref{eq:salza-cond}), the two-terms factorization elements of $S_f(\givenany{x}{y})$ can be
interpreted the following way:
\begin{enumerate}
\item ${\mathcal S}$ is based on the length ratio of $x$ that is explained by $y$ -- we will show
  that it acts as a ``spreading'' factor that emphasizes differences between both strings so that the
  final value allows for a sharper numerical estimate (see Sec.~\ref{ssec:salza:study});
\item ${\mathcal Z}$ is the normalization of the \salza{} approximation of the relative complexity
  \cite{ziv:relative}. This normalization is simply obtained by dividing by the maximum number of
  symbols that can be produced, namely $|x|$.
\end{enumerate}

\begin{lemma}\label{prop:normalized}$0\leq S_f(\givenany{x}{y})< 1.$
\end{lemma}
\begin{proof}
  The proof consists in demonstrating that both ${\mathcal S}$ and ${\mathcal Z}$ are normalized and follow the same trend with respect to
  similarity.
  \begin{enumerate}
  \item First, we show that the term ${\mathcal S}$ is normalized.
    
    Remark that:
    $$\sum_{\allrefs{\givenany{x}{y}}}l f(l)-|\allrefs{\givenany{x}{y}}|_f = \sum_{\allrefs{\givenany{x}{y}}}f(l)(l-1).$$
%    \begin{itemize}
    %    \item
    
    Because $f$ is upper-bounded by 1 and remembering that $l\geq 1$, see
    Eq.~(\ref{eq:lengths:all:geq1}), then:
    $$\sum_{\allrefs{\givenany{x}{y}}}f(l)(l-1)\leq \sum_{\allrefs{\givenany{x}{y}}}(l-1),$$
    and:
    $$\sum_{\allrefs{\givenany{x}{y}}}(l-1)\leq |x|-1.$$
    Note that the equality holds when $x$ is a substring of $\mathcal{R}$ (\textit{e.g.}, when
    computing $S_f(\given{x}{}{+}{x})$).
    
    Hence:
    $$\left(1-\frac{\sum_{\allrefs{\givenany{x}{y}}}l f(l)-(|\allrefs{\givenany{x}{y}}|_f-1)}{\left|x\right|}\right)\geq 0,$$
    with zero being reached when $x$ is a substring of $\mathcal{R}$. % (\textit{e.g.}, when computing
    %      $S_f(\given{x}{}{+}{x})$ in Ziv-Merhav mode).
    %    \item
    
    By Eq.~(\ref{eq:lengths:all:geq1}) and positivity of $f$, one has:
    $$\sum_{\allrefs{\givenany{x}{y}}}f(l)(l-1)\geq 0.$$
    Therefore:
    $$\left(1-\frac{\sum_{\allrefs{\givenany{x}{y}}}l f(l)-(|\allrefs{\givenany{x}{y}}|_f-1)}{\left|x\right|}\right)\leq 1,$$
    with equality being reached when $x$ and $\mathcal{R}$ are maximally dissimilar (\textit{e.g.},
    $\alphabet_x\cap\alphabet_y=\emptyset$ when computing $S_f(\given{x}{}{+}{y})$).
    %    \end{itemize}
  \item Second, we show that ${\mathcal Z}$ is also normalized.
    %    \begin{itemize}
    %    \item
    
    The \salza{} factorization will produce at least one symbol (exactly one when $x$ is a substring of
    $\mathcal{R}$), therefore:
    $$\frac{|\allrefs{\givenany{x}{y}}|-1}{|x|}\geq 0.$$
    Note that ${\mathcal Z}$ only vanishes when $x=y$.
    %    \item

    In the worst case (when $x$ and $\mathcal{R}$ are maximally dissimilar, \textit{e.g.}
    $\alphabet_x\cap\alphabet_\mathcal{R}=\emptyset$, same as above), $|\allrefs{\givenany{x}{y}}|=|x|$.
    Therefore:
    $$\frac{|\allrefs{\givenany{x}{y}}|-1}{|x|} = \frac{|x|-1}{|x|}<1.$$
%    \end{itemize}
  \end{enumerate}
  From the above considerations, it is easy to see that, when $x$ and $y$ are maximally dissimilar, one
  has:
  $$\lim_{|x|\rightarrow\infty}S_f(\givenany{x}{y}) = 1.$$
\end{proof}

Before making use of $S_f(\givenany{x}{y})$ to measure complexity, we investigate some of its features.

\subsection{Study of the \salza{} conditional complexity estimate}\label{ssec:salza:study}

We start by choosing an admissible function $f$. The use of $f$ in Eq.~(\ref{eq:salza-cond}) allows to
modulate the choice of the references taken into account, and how they contribute to the construction
of the estimate of the complexity.

A possible choice for $f$ is a threshold function in order to filter out references that are not
meaningful. Such references are first defined as explained below.

\begin{definition} Meaningful references \cite{lempel:complexity}.

  A \salza{} reference $\left(l,v\right)$ is said to be meaningful with respect to $\mathcal{R}$ iff:
  \begin{equation}
    l > l^0_\mathcal{R} = \minlen{\mathcal{R}}{}.
  \end{equation}
\end{definition}

Then, an admissible threshold function is defined as follows:

\begin{definition} Threshold function.

  The admissible threshold function for string $\mathcal{R}$, denoted $f_\mathcal{R}^t$, is defined as:
  \begin{align*}
    f_\mathcal{R}^t(l)=\left\{
    \begin{aligned}
      1 & & \mbox{if} & & l > l^0_\mathcal{R} \\
      0 & & \mbox{otherwise} &
    \end{aligned}
    \right..
  \end{align*}
\end{definition}

However, as expected, this choice produces discontinuities in the first term of
$S_{f_\mathcal{R}^t}(\givenany{x}{y})$, {\em i.e.} ${\mathcal S}_\mathcal{R}^t$, see
Fig.~(\ref{fig:admissible:threshold}). The effect of using this threshold is especially harmful
when $l_\mathcal{R}^0$ is close to $\lfloor l_\mathcal{R}^0\rfloor$: in that case, references which
are just smaller than $l_\mathcal{R}^0$ are discarded, even though they may carry some information. 

An easy way to circumvent this issue is to replace the threshold function with a continuous function in
order to produce soft estimates. Among all possible choices, we arbitrarily favor $C^\infty$ functions
and make use of a sigmoid.

\begin{definition} Sigmoid function.

  The admissible sigmoid function for string $\mathcal{R}$, denoted $f_\mathcal{R}^s$, is defined as:
  \begin{equation*}
    f_\mathcal{R}^s(l) = \frac{1}{1+e^{-l+l_\mathcal{R}^0}}.
  \end{equation*}
\end{definition}

\begin{figure}[!ht]
\centering
\includegraphics[width=.5\linewidth]
        {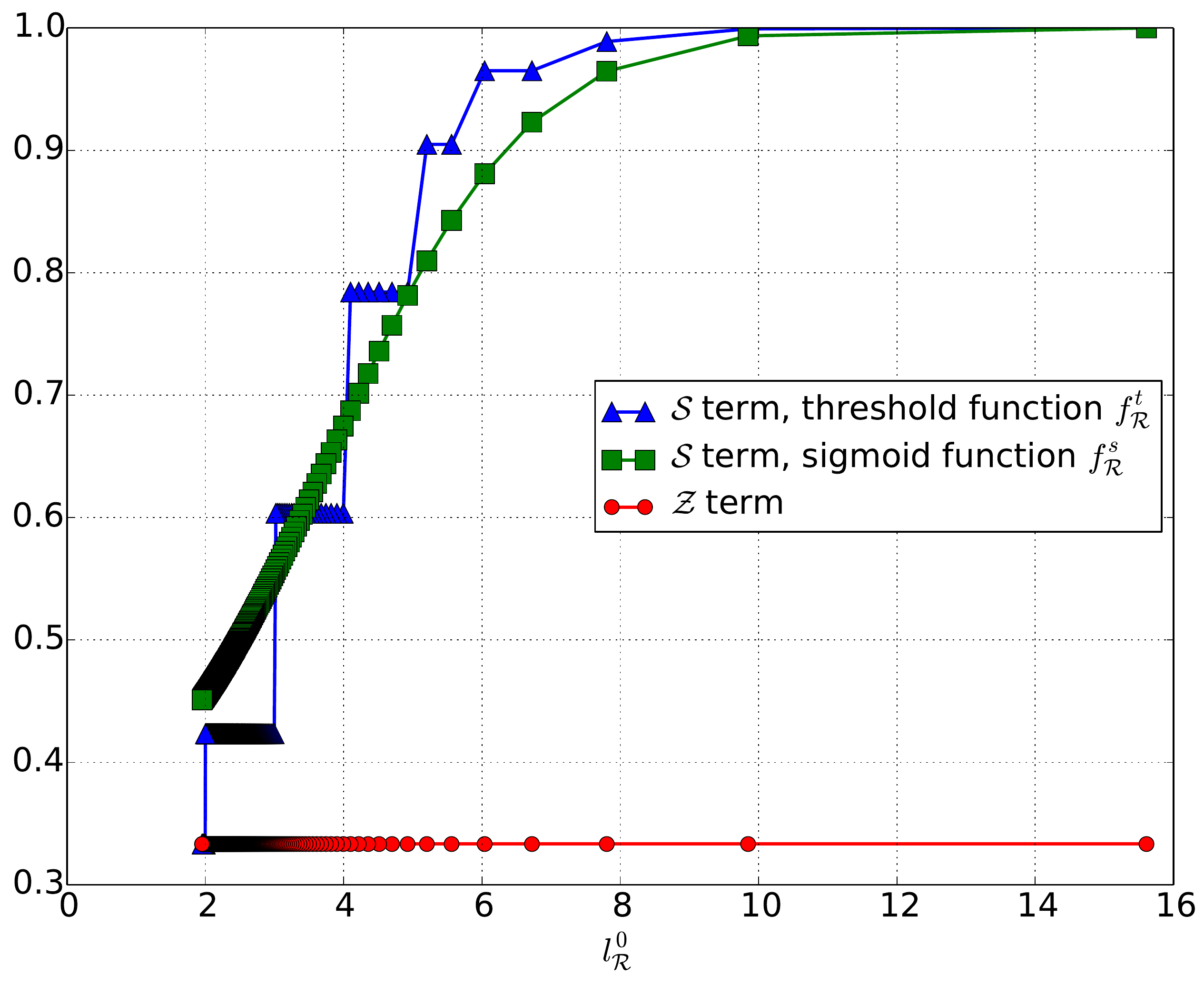}
        \caption{Profiles of $S_{f_\mathcal{R}^t}(\givenany{x}{y})$ (triangles) vs.
          $S_{f_\mathcal{R}^s}(\givenany{x}{y})$ (squares) as a function of $l_\mathcal{R}^0$. Using a
          continuous admissible function like a sigmoid allows to circumvent the effect of taking hard
          decisions when thresholding. This is especially important for small references and helps better
          discriminate dissimilar strings. }
\label{fig:admissible:threshold}
\end{figure}

In the following, we study the evolution of the various scores with respect to $\mu$ the mean length
of the symbols produced during the factorization, $|x|$ and $l_\mathcal{R}^0=\minlen{\mathcal{R}}{}$.
This study is hardly tractable using closed-form expressions and we resort to numerical simulations
(this is inspired by the probabilistic treatment found in \cite{ziv:relative}). From our experiments,
we chose the Poisson distribution as the most suitable discrete distribution for the symbol lengths.

In order to study $S_f(\givenany{x}{y})$ with respect to $\mu$, we have fixed $l_\mathcal{R}^0$ and
generated Poisson-distributed lengths such that $\sum_{\allrefs{\givenany{x}{y}}}l=|x|$. The results are
depicted in Fig.~(\ref{fig:boost:small:values}): when $\mu$ takes high values (\textit{i.e.}, $x$ and
$\mathcal{R}$ contain many identical long strings) all three terms get the same. However, when $\mu$ is
smaller (\textit{i.e.}, $x$ and $\mathcal{R}$ only contain short identical strings), then the spreading
term $\mathcal{S}$ of \salza{} in Eq.~(\ref{eq:salza-cond}) will take much higher values than the term
$\mathcal{Z}$ alone. Hence, the discriminative power of \salza{} for dissimilar strings is expected to be
much stronger. 

In order to study $S_f(\givenany{x}{y})$ with respect to $|x|$, we have worked with both $\mu$ and
$l_\mathcal{R}^0$ fixed. The results are depicted in Fig.~(\ref{fig:salza:long:strings}): they confirm
that the \salza{} conditional complexity estimate is invariant with respect to $x$. 

In Fig.~(\ref{fig:admissible:threshold}), it is noticeable that the term $\mathcal{Z}$ is constant with
respect to $l_\mathcal{R}^0$. This is actually an issue. Suppose $\mathcal{R}$ is a long string defined
over a small alphabet, then it is very likely that $\mathcal{R}$ will contain almost all possible
strings of small size. Thus, if $x$ is defined over the same small alphabet, there is a high chance that
it will be explained using only small strings found in $\mathcal{R}$. Therefore, it seems important to
take into account both the values of $|\mathcal{R}|$ and $|\alphabet_\mathcal{R}|$ and this can be done
using the minimal length of meaningful references $l_\mathcal{R}^0$. 

Now, as opposed to the $\mathcal{Z}$ term, the ``spreading'' term of Eq.~(\ref{eq:salza-cond}),
{\em i.e.} ${\mathcal Z}$, will adapt to all settings. Using a sigmoid allows for a softer estimate than
when using a threshold. This, in fact, contributes to the stronger discriminative power of \salza{}
when common substrings are close to $l_\mathcal{R}^0$, see Fig.~(\ref{fig:boost:small:values}). 

\begin{figure}[!ht]
\centering
\includegraphics[width=.5\linewidth]
        {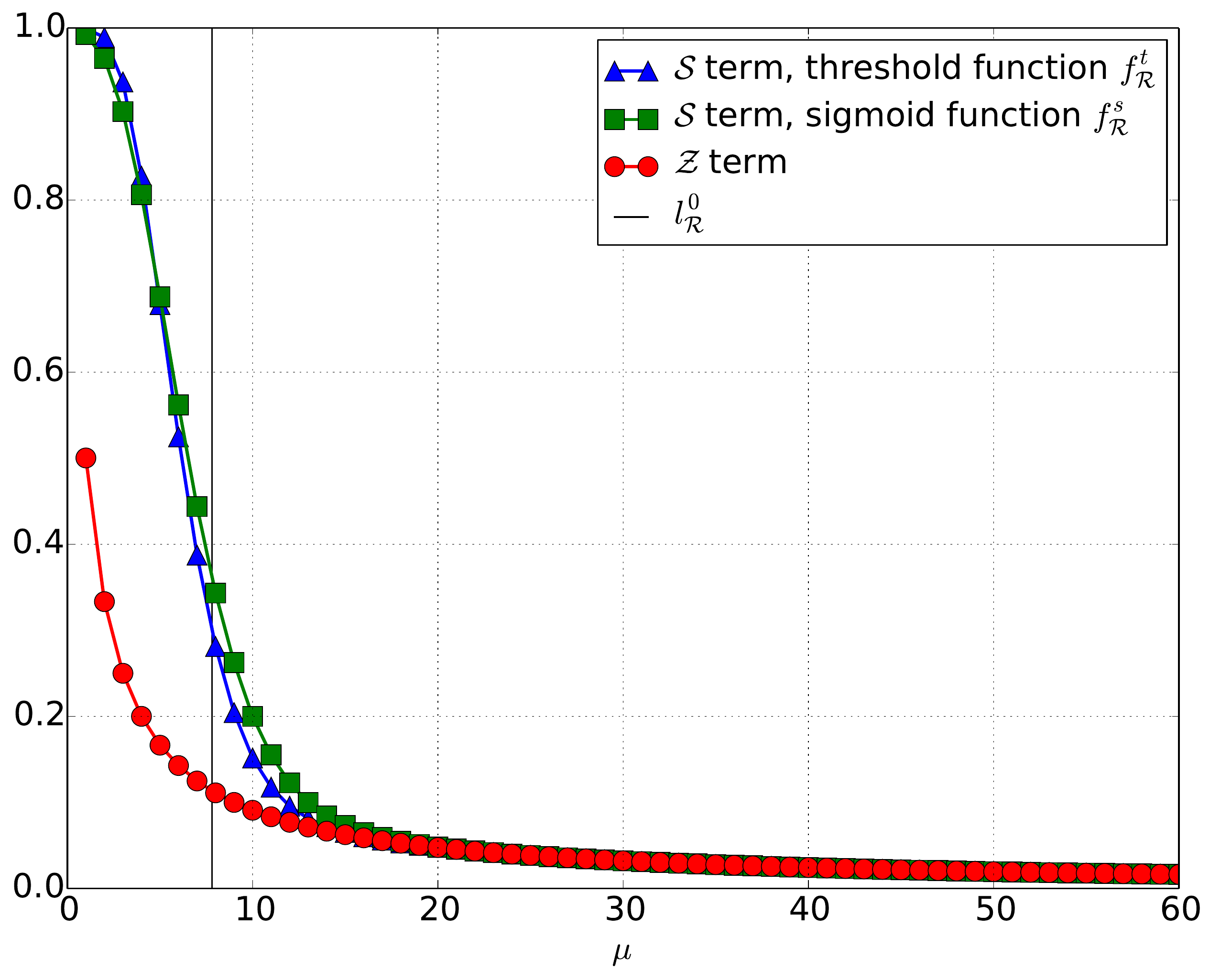}
        \caption{Discriminative power of \salza{} as a function of the Poisson distribution mean 
          $\mu$. The minimal value of the meaningful references $l_\mathcal{R}^0$ is represented as a
          vertical bar. }
\label{fig:boost:small:values}
\end{figure}

\begin{figure}[!ht]
\centering
\includegraphics[width=.5\linewidth]
        {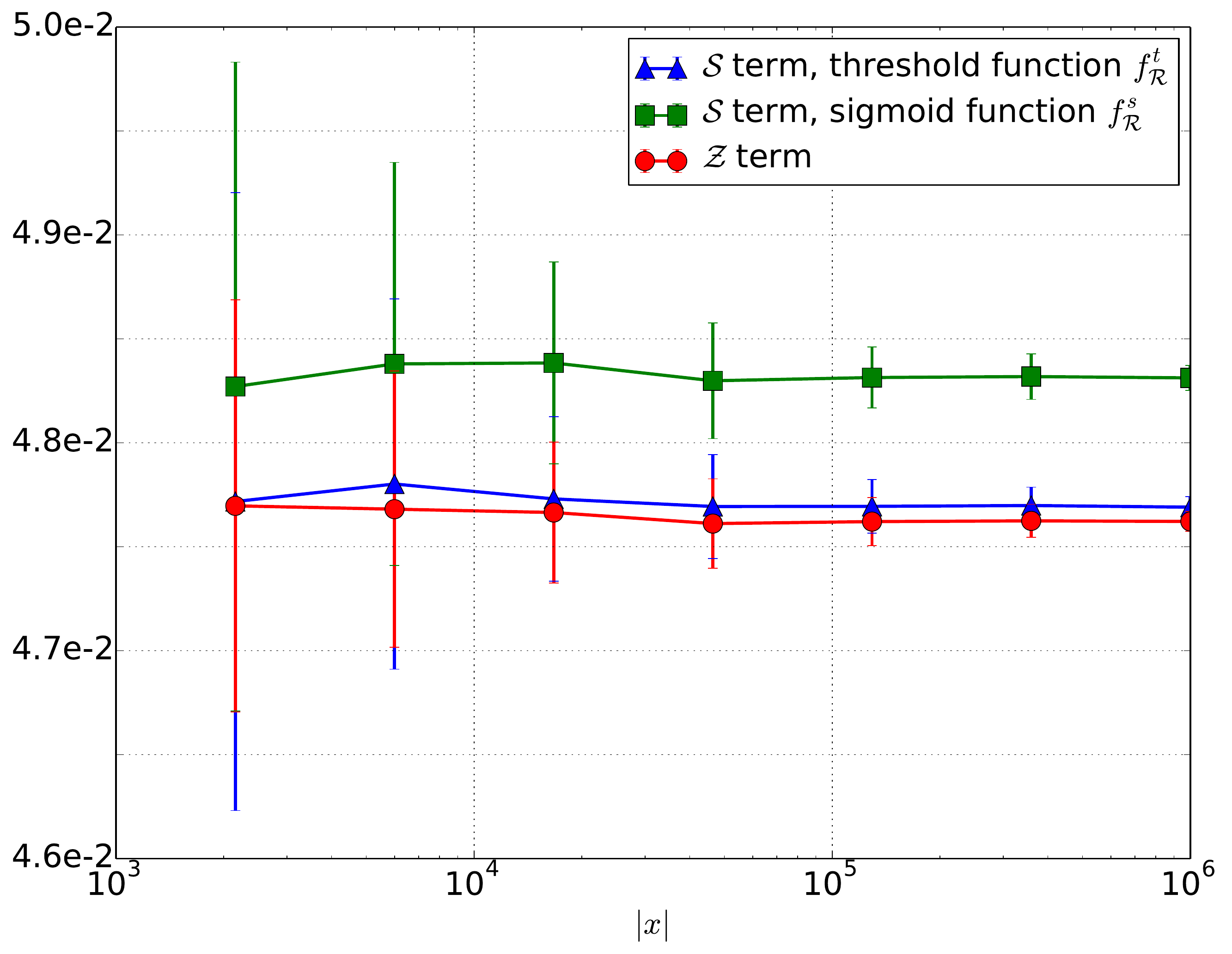}
        \caption{$S_{f_\mathcal{R}^t}(\givenany{x}{y})$ (triangles) and
          $S_{f_\mathcal{R}^s}(\givenany{x}{y})$ (squares) as a function of $|x|$. The discriminative
          power of \salza{} is nearly constant for arbitrarily long strings -- this is in contrast with
          the NCD (see Sec.~\ref{ssec:app:clustering}) and it is due to the unbounded buffer in our
          implementation. Note that the plot starts when $|x|$ is half the size of the \texttt{DEFLATE}
          buffer. Each point has been generated by averaging 200 realizations. }
\label{fig:salza:long:strings}
\end{figure}

Overall, the \salza{} conditional complexity estimate will better discriminate dissimilar strings than
the $\mathcal{Z}$ term alone, and it will do so equally well on arbitrarily long strings. As for the
choice of the admissible function, we observe a little advantage in favor of the sigmoid. Therefore, in
the rest of this paper, unless otherwise explicitly stated, the sigmoid is the default admissible
function for all computations. 

\subsection{\salza{} complexity and joint complexity}\label{ssec:salza-joint}

Let $\allrefs{x}=\allrefs{\given{x}{}{}{x}}$ be the set of lengths produced during a regular
\texttt{DEFLATE} factorization (\textit{i.e.}, a close version of that in \cite{ziv:universal}). This
allows us to propose an estimate for the joint complexity of strings $x$ and $y$. 

\begin{definition} \salza{} complexity.
  
  Given an admissible funtion $f$ and a non-empty string $x\in\alphabet_x$, the \salza{} complexity of
  $x$, denoted $S_f(x)$, is defined as:
  \begin{align*}
    S_f(x) & = S_f(\given{x}{}{}{x}) \\
    & = \left(1-\frac{\sum_{\allrefs{x}}l f(l)-(|\allrefs{x}|_f-1)}{\left|x\right|}\right) \frac{|\allrefs{x}|-1}{|x|}.
  \end{align*}
\end{definition}

The joint Kolmogorov complexity can be understood as the minimal program length able to encode both
$x$ and $y$, as well as a means to separate the two \cite{li:kolmogorov}. Hence, there is no need to
restrict the references only to $x$, and we should allow references to the past of $y$ as well. In order
to mimic the relationship $K(x,x)=K(x)$, we choose a length ratio as the way to separate both strings. 

\begin{definition} \salza{} joint complexity.
  
  Given an admissible function $f$, and two non-empty strings $x\in\alphabet_x$ and $y\in\alphabet_y$,
  the \salza{} joint complexity of $x$ and $y$, denoted $S_f(x,y)$, is defined as:
  \begin{equation*}
    S_f(x,y) = S_f(\given{y}{-}{+}{x})+S_f(x)+\log_{|\alphabet_x|}\left(\frac{|x|}{|y|}\right).
  \end{equation*}
\end{definition}

Note that $S_f(x,x) = S_f(x)$ because $S_f(\given{x}{-}{+}{x})=0$. 

In order to validate our approach, we measure the following absolute error:
$$\epsilon=|S_f(x,y)-S_f(y,x)|.$$

For each experiment, we have used various translations for the Universal Declaration of Human Rights
(UDHR) and samples of mammals mitochondrial DNA (available at \cite{cilibrasi:complearn}). The results
are reported in Tab.~\ref{tab:joint:symmetry} and show a maximum average absolute error below 2.37\%,
while the maximum absolute error we have encountered is around 7.2\% (which happened in the weirdest
setting: comparing completely unrelated data, namely a DNA sample and a human text -- when the data
come from the same area, the results are much better on average). 

\begin{table}[!ht]
  \caption{Characterization of \salza{} joint complexity with respect to symmetry.}
  \label{tab:joint:symmetry}
  \centering
  \begin{tabular}{|l|l||l|l|l|l|}\hline
    x    & y    & $\mathbb{E}\left[\epsilon\right]$ & $\mbox{Var}\left[\epsilon\right]$ & $\mbox{min}(\epsilon)$ & $\mbox{max}(\epsilon)$ \\
    \hline
    UDHR & UDHR & 1.43e-3 & 1.38e-6 & 5e-6    & 7.96e-3 \\
    DNA  & DNA  & 1.23e-3 & 8.11e-7 & 6e-6    & 4.98e-3 \\
    UDHR & DNA  & 6.84e-2 & 2.49e-6 & 6.28e-2 & 7.17e-2 \\
    \hline
  \end{tabular}
\end{table}

\section{\NSD{} and directed information}\label{sec:dist:dirinfo}

In this section, we use the previous definitions to devise both an algorithmic semi-distance and
directed information definitions that are key to the applications in Sec.~\ref{sec:salza:app}.

\subsection{The normalized \salza{} semi-distance}\label{ssec:salza-semi-metric}

Using \salza{} in its Ziv-Merhav mode practically eliminates the two limitations mentioned in
Sec.~\ref{sec:intro}: references are taken only from $\mathcal{R}=y$, and the unbounded buffer allows to
reach any arbitrary point in $y$.

Since $S_f(\given{x}{}{+}{y})$ is normalized, we can now refer directly to
Eq.~(\ref{eq:information:distance}) to propose a semi-distance. 

\begin{definition}\label{def:nssm} \NSD{f}.
  
  Given an admissible function $f$, and two non-empty strings $x\in\alphabet_x$ and $y\in\alphabet_y$,
  the normalized \salza{} semi-distance, denoted \NSD{f}, is defined as:
  \begin{equation*}
    \NSD{f}(x,y) = \mbox{max}\left\{S_f(\given{x}{}{+}{y}), S_f(\given{y}{}{+}{x})\right\}.
  \end{equation*}
\end{definition}

Note that \NSD{f} stands for Normalized \salza{} semi-Distance using $f$. By default, when $f=f^s_y$,
it is simply denoted by \NSD{}.

\begin{thm}
  \NSD{f} is a normalized semi-distance.
\end{thm}
\begin{proof}
  By Lemma \ref{prop:normalized} and Def.~\ref{def:nssm}, \NSD{f} is normalized. We now show
  that \NSD{f} is a semi-distance:

%  \begin{enumerate}
%\item
  $\NSD{f}(x,y)=\NSD{f}(y,x)$ (symmetry): This is immediate by Def.~\ref{def:nssm}.
    %  \item
    
  $\NSD{f}(x,y)\geq 0$ (non-negativity): This is immediate by Lemma \ref{prop:normalized} and
  Def.~\ref{def:nssm}.
  
%\item
  $\NSD{f}(x,y) = 0 \iff x=y$ (identity of indiscernibles):
    \begin{itemize}
    \item $x=y \implies \NSD{f}(x,y) = 0$

      When $x=y$, then \salza{} will produce only one reference to the entire string $y$ (resp. $x$)
      when computing $S_f(\given{x}{}{+}{y})$ (resp. $S_f(\given{y}{}{+}{x})$). Therefore:
      $$x=y \implies S_f(\given{x}{}{+}{y})=S_f(\given{y}{}{+}{x})=0.$$
    \item $\NSD{f}(x,y) = 0\implies x=y$

      Since $S_f(\given{x}{}{+}{y})\geq 0$, then by Def.~\ref{def:nssm}:
      $$\NSD{f}(x,y)=0 \implies S_f(\given{x}{}{+}{y})=S_f(\given{y}{}{+}{x})=0.$$
      Because it is a product, setting $S_f(\given{x}{}{+}{y})=0$ means either (or both) of its terms is
      zero. Looking at the demonstration of Lemma~\ref{prop:normalized}, one sees that setting either
      member to zero means $x$ is a substring of $y$ (when computing $S(\given{x}{}{+}{y})$) and $y$
      is a substring of $x$ (when computing $S(\given{y}{}{+}{x})$). Hence, $x=y$. This completes the
      proof for the identity of indiscernibles.
    \end{itemize}
    %  \item
    
    \noindent Showing that the triangle inequality does not hold only requires a counter-example. 
    Consider a string $x$. Let $\bar{x}$ be $x$ with its literals in reversed order and $x^n$ be $n$
    times the concatenation of $x$. Let $a\in\alphabet_a$, $b\in\alphabet_b$ and $c\in\alphabet_c$ and
    let all strings $a$, $b$
    and $c$ be each composed of the concatenation of all literals -- all three strings have the same
    length $|\alphabet_a|=|\alphabet_b|=|\alphabet_c|=M$. Assume further that
    $\alphabet_a\cap\alphabet_b=\alphabet_b\cap\alphabet_c=\alphabet_a\cap\alphabet_c=\emptyset$.
    Now, by concatenation, define $x=a^nb\bar{c}$, $y=bc^n\bar{a}$ and $z=acb^n$. Let also
    $f = f_y^t$.
    In this case:
    \begin{itemize}
    \item $\NSD{f}(x,y)=\frac{(n+1)^2}{(n+2)^2}$, since both computations of $S_f(\given{x}{}{+}{y})$
      and $S_f(\given{y}{}{+}{x})$ will generate $M(n+1)$ literals and only one reference (of lengh $M$);
    \item $\NSD{f}(x,z)=\frac{(n+M)^2}{(n+2)^2M^2}$, since both computations of
      $S_f(\given{x}{}{+}{z})$ and $S_f(\given{z}{}{+}{x})$ will generate $M$ literals and $n+1$
      references (of length $M$);
      \item $\NSD{f}(z,y)=\frac{(n+M)^2}{(n+2)^2M^2}$ (following the same reasoning as above).
    \end{itemize}
    Let now $T=\NSD{f}(x,z)+\NSD{f}(z,y)-\NSD{f}(x,y)$. One has:
    $$T=\frac{2(n+M)^2-(n+1)^2M^2}{(n+2)^2M^2}.$$
    For example, when $M=60$ and $n=10$, one gets $T=-0.54<0$, which violates the triangle inequality. 
%  \end{enumerate}
\end{proof}

The counter-example above basically shows that coders based on \cite{ziv:universal} or
\cite{ziv:relative} will not handle correctly anti-causal phenomenons, although one may of course
design an extended lookup procedure able to handle such cases in addition to the usual ``copy from the
past'' used here. For this reason, we believe tuning the admissible function $f$ will not help in
restoring the triangle inequality. 

In addition, below are a couple of facts about of the \mbox{NCD}:
\begin{itemize}
\item when the \mbox{NCD} is computed using real-world compressors, the triangle inequality may also
  be violated in practice (one may devise a counter-example using the same reasoning as above that
  will also violate the triangle inequality for NCD/\texttt{gzip});
\item in all experiments in this paper, we did not witness any violation of the triangle inequality
  when using \salza{};
\item the other properties hold strictly, which may not be the case for the practical versions of the
  NCD (\emph{e.g.}, for large $x$, NCD/\texttt{gzip}$(x,x)$ will not be zero);
\item the computational cost of \salza{} is different compared to the NCD: when using \texttt{gzip},
  it is bounded by the internal buffer of 32KiB. However, since we use an unbounded buffer, the
  \salza{} running complexity is $O(\mbox{max}(|x|, |y|))$. Therefore, only running two coders instead
  of three in Eq.~(\ref{eq:ncd}) does not necessarily implies faster running times for \salza{}. 
\end{itemize}

\subsection{Directed algorithmic information}\label{ssec:dirinfo}

Causality inference relies on the assessment of a matrix of directed informations from which a causality
graph will be produced. Due to the very nature of causality, some fundamental restrictions on the
underlying graph structure apply. In particular, most authors focus on directed acyclic graphs (DAG)
\cite{pearl:causality} and we will hereafter follow this line. Therefore, we start by defining estimates
of directed algorithmic information. 

We would like to stress that causality has received several interpretations and it is, among other
considerations, also dependent on the type of data at hand. We will consider two types of data here:
time series \cite{granger:causal} (in which application area a version based on classical information
theory has been proposed \cite{amblard:granger:review}), and data that is not a function of time
\cite{pearl:causality}. To some extent, this relates to the difference between online and offline
applications. Therefore, we need to distinguish between the two. 

Let $X=\left\{x_i\right\}$ be a set of strings, and let us denote $X\backslash Y$ the set from which the
set of strings $Y$ was removed ($Y\subset X$). When $Y=\left\{y\right\}$, we also write $X\backslash y$. 

We formulate the causal directed algorithmic information as follows:

\begin{definition} Causal directed algorithmic information.  
\begin{equation}
  \forall i\neq j,\;C(x_i\rightarrow x_j) = K( \given{x_j}{-}{}{X\backslash \left\{x_i,x_j\right\}} ) -
  K(\given{x_j}{-}{}{X}\backslash x_j).
\end{equation}
\end{definition}

$C(x_i\rightarrow x_j)$ is the amount of algorithmic information flowing from $x_i$ to $x_j$ when
oberving data online in real time (think of the $x_i$ as {\em e.g.}, outputs of ECG probes). 

In practice, we compute:
\begin{equation}\label{eq:causal:directed:info}
  C_{S_f}(x_i\rightarrow x_j) = S_f(\given{x_j}{-}{}{X\backslash \left\{x_i, x_j\right\}} ) - S_f(\given{x_j}{-}{}{X}\backslash x_j).
\end{equation}

Similarly, for offline applications, when all the data is available beforehand (think {\em e.g.}, of
text excerpts), we define the socalled full directed algorithmic information as: 

\begin{definition} Full directed algorithmic information.  
\begin{equation}
  \forall i\neq j,\;F(x_i\rightarrow x_j) = K( \given{x_j}{-}{+}{X\backslash \left\{x_i,x_j\right\}} ) -
  K(\given{x_j}{-}{+}{X\backslash x_j}).
\end{equation}
\end{definition}

In practice, we compute:
\begin{equation}\label{eq:full:directed:info}
  F_{S_f}(x_i\rightarrow x_j) = S_f(\given{x_j}{-}{+}{X\backslash \left\{x_i,x_j\right\}} ) - S_f(\given{x_j}{-}{+}{X\backslash x_j}).
\end{equation}

Note that we are only considering the amount of information flowing from one string to another. Hence,
we are fundamentally fitting in the Markovian framework. And since we remove the influence of all other
strings, we are actually measuring the influence of the sole innovation contained in one such string onto
another. This will be illustrated in Sec~\ref{ssec:app:causality}. 

%In the last equations involving \salza{}, we simply concatenate the required strings to form the
%\textit{a priori}, conditional data $\mathcal{R}$.

\section{Applications}\label{sec:salza:app}

\subsection{Application to clustering}\label{ssec:app:clustering}

In Eq.~(\ref{eq:salza-cond}), we use the richness of the information that is produced when running a
Ziv-Merhav coder \cite{ziv:relative}. Such a symbol-length information is not available when using a
real-world compressor out of the box. In this subsection, we show that our strategy allows for a more
sensitive assessment of the complexity compared to the \texttt{CompLearn} tool
\cite{cilibrasi:complearn}. Throughout this paper, we make use of the Neighbor-Joining (NJ) method to
obtain the phylogenic trees \cite{saitou:neighbor:joining} every time we need to compare against the
state-of-the-art, although the UPGMA method \cite{sokal:upgma} would produce essentially the same
results (though it will be used appropriately in Sec.~\ref{sssec:app:toussaint}). 

\subsubsection{Small, real data}%\label{ssec:realdata}

We have used the data available at \cite{cilibrasi:complearn} to compare the phylogenic trees produced
by \salza{} and those produced with \texttt{CompLearn} \cite{cilibrasi:complearn} using \texttt{gzip}
for the \mbox{NCD} compressor\footnote{Although we have developed our own \texttt{gzip} compliant
  encoder, including the Huffman entropy coding stage, we have used the Gailly \& Adler implementation
  for fairness of the simulations (available at \href{http://www.gzip.org}{\texttt{gzip.org}}).}. These
datasets are made of mitochondrial DNA samples and several translations of the Universal Declaration of
Human Rights in various languages (more properly, writing systems). 

\begin{figure*}%
  \centering
  \begin{minipage}[b]{.5\linewidth}
    \includegraphics[width=\textwidth,height=11cm]
        {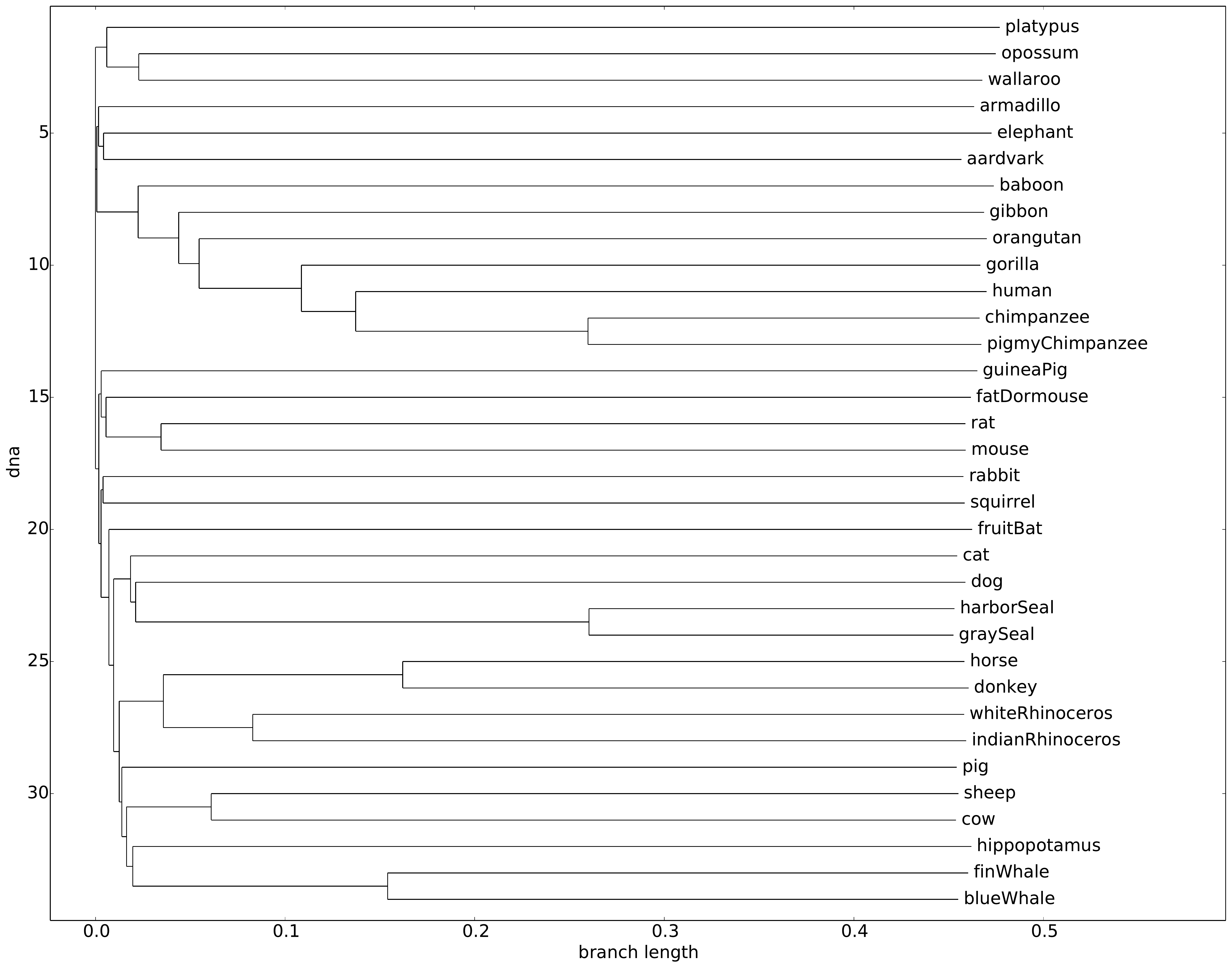}
  \end{minipage}%
  \hfil
  \begin{minipage}[b]{.5\linewidth}
    \includegraphics[width=\textwidth,height=11cm]
        {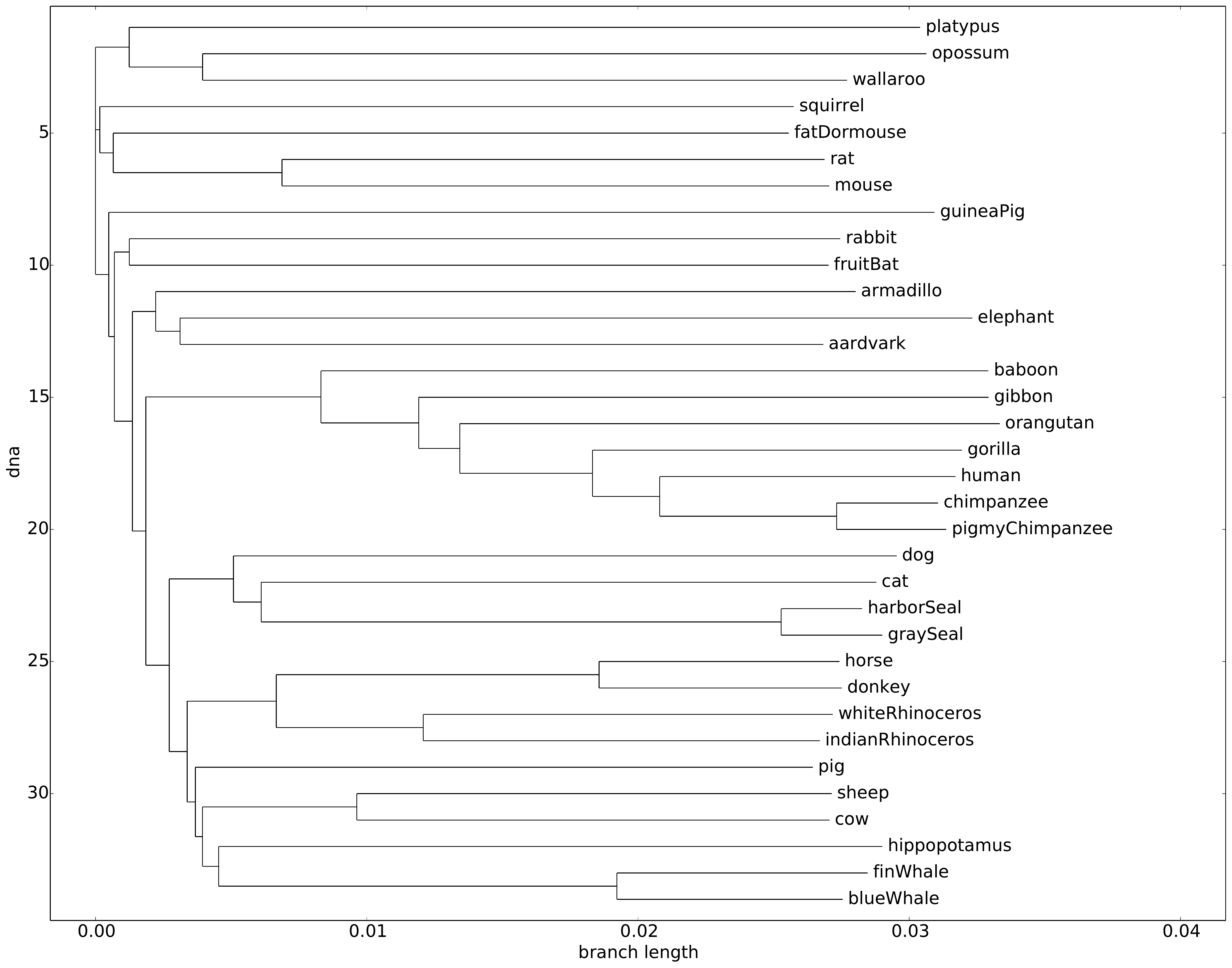}
  \end{minipage}
  \caption{Mitochondrial DNA clustering with NCD/\texttt{gzip} (left) vs. \NSD{} (right).}
  \label{fig:dna}
\end{figure*}

\begin{figure*}%
  \centering
  \begin{minipage}[b]{.5\linewidth}
    \includegraphics[width=\textwidth,height=11cm]
        {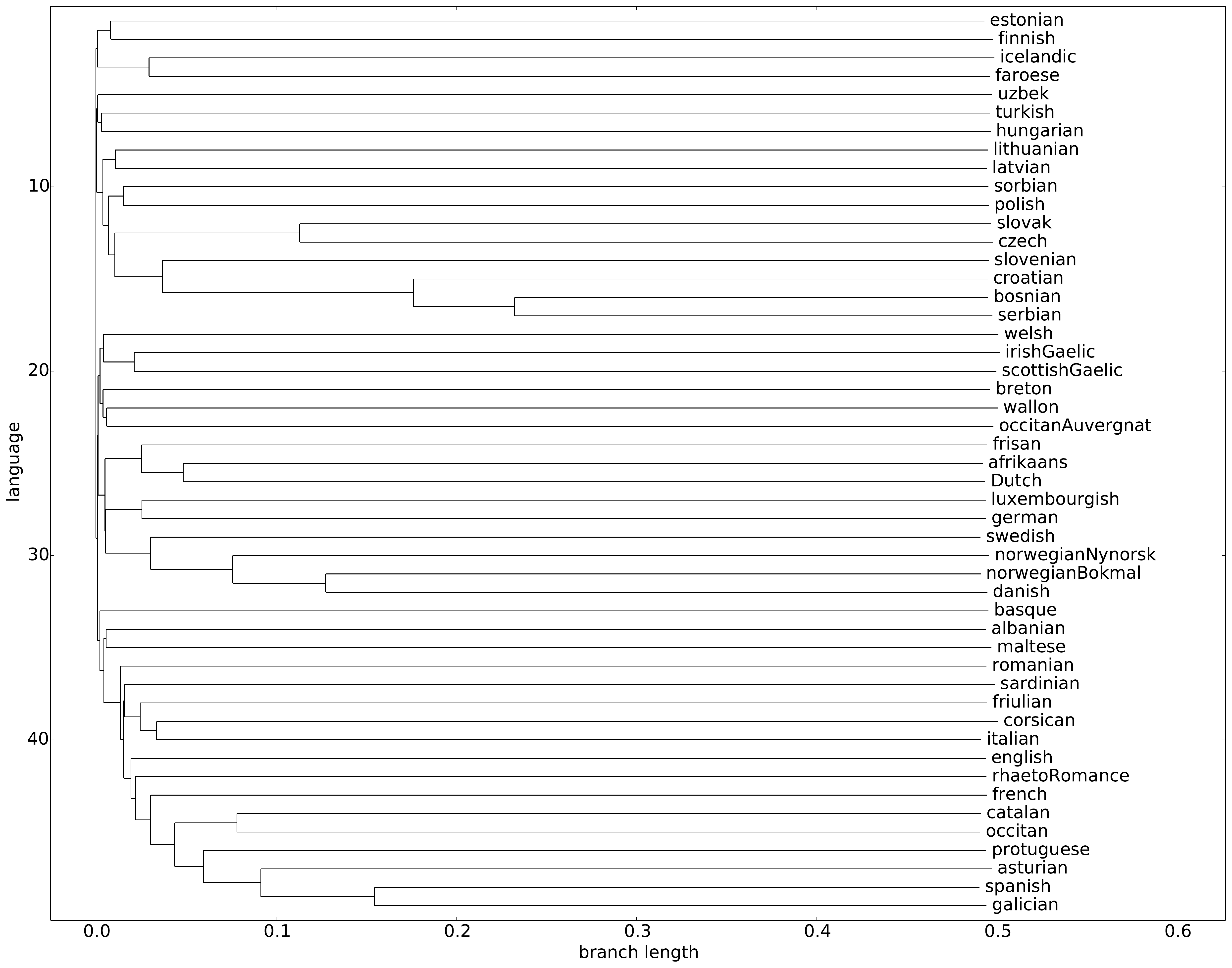}
  \end{minipage}%
  \hfil
  \begin{minipage}[b]{.5\linewidth}
    \includegraphics[width=\textwidth,height=11cm]
        {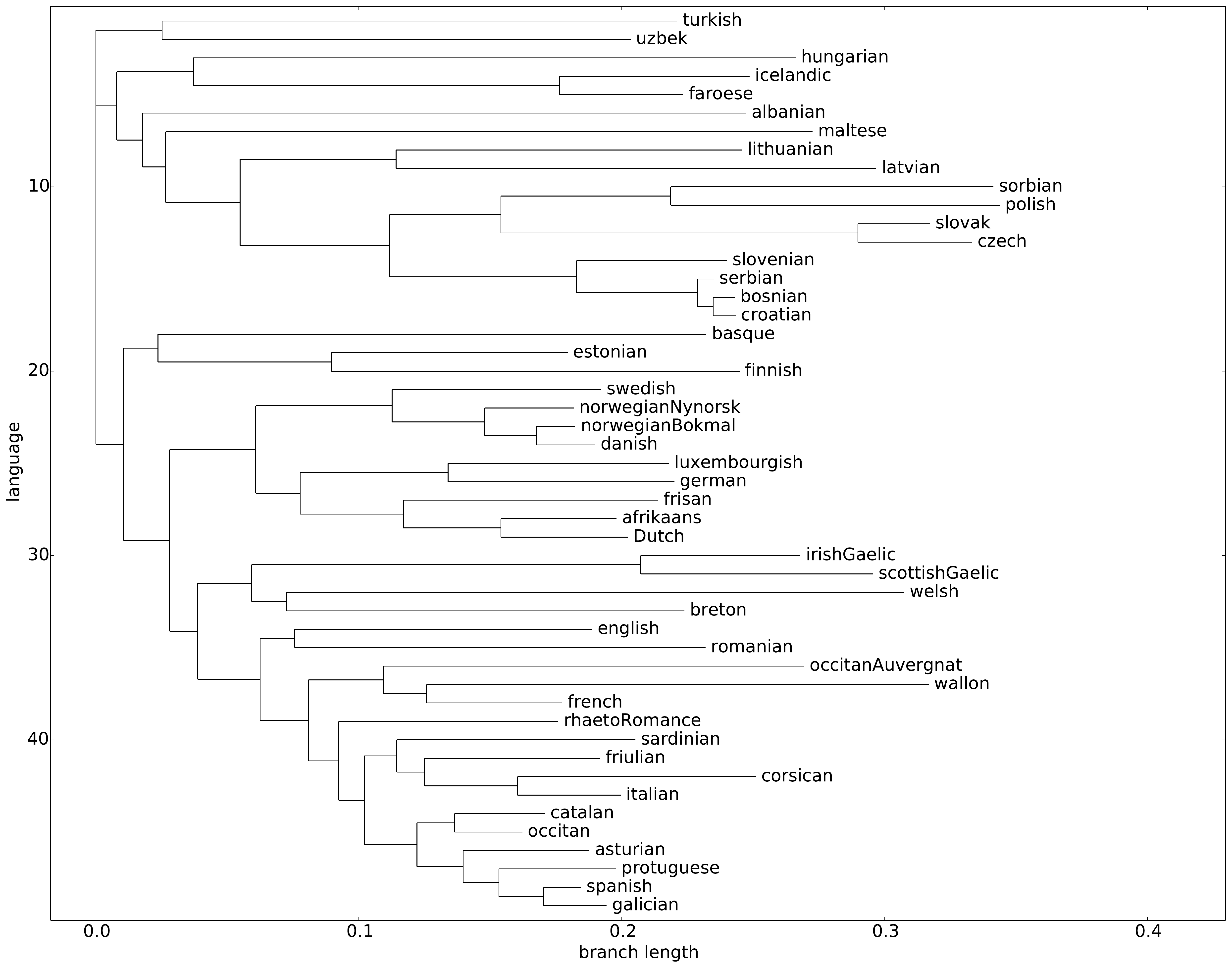}
  \end{minipage}
  \caption{Clustering of writing systems (UDHR) using NCD/\texttt{gzip} (left) vs. \NSD{} (right).}
  \label{fig:ddhc}
\end{figure*}

The effect of the first term in Eq.~(\ref{eq:salza-cond}) is clear: the trees exhibit sharper clustering
patterns than when using NCD/\texttt{gzip}, see Fig.~(\ref{fig:dna}--\ref{fig:ddhc}). One has obviously
a real advantage in taking into account the lengths produced by the \salza{} factorization. This
advantage seems even clearer when the data is more structured as it is the case with human writing
systems. 

It is remarkable that regarding the Basque language, NCD/\texttt{gzip} fails to group it with
Finno-Ugric languages (Finnish and Estonian), which is one of the earliest hypotheses
(XIX\textsuperscript{th} century) on the question of the origin of the Basque language
(\cite{morvan:basque}, Ch. 2). This hypothesis has been mostly rejected in the meantime. However,
the stronger discriminative power of \salza{} shows that this hypothesis could have seem plausible
at some point in time. Because the branches of the phylogenic trees can be put upside-down without
changing the meaning of the interpretation, an acute reader will see that the Basque language is also
close to Altaic languages (Turkish, Uzbek), which count among other hypotheses that have been formulated
in the past \cite{morvan:basque}.

\subsubsection{Synthetic data}%\label{ssec:salza-sim-markov}

As seen above, the phylogenic trees show sharper differences between clusters with \salza{}. We now use
several realizations of various first-order Markovian processes to compare with the \mbox{NCD} in the
same setting. In the following plots, the data is labeled as \texttt{alphaX\_mY\_cZ}, where
$\mbox{\texttt{X}}=|\alphabet|$, \texttt{Y} is the identifier of the Markov transition probabilities
matrix and \texttt{Z} is the identifier of the realization. Each realization string contains 15K
literals (actually, bytes), so that we can fairly compare with NCD/\texttt{gzip}. 

\begin{figure*}%
  \centering
  \begin{minipage}[b]{.5\linewidth}
    \includegraphics[width=\textwidth,height=11cm]
        {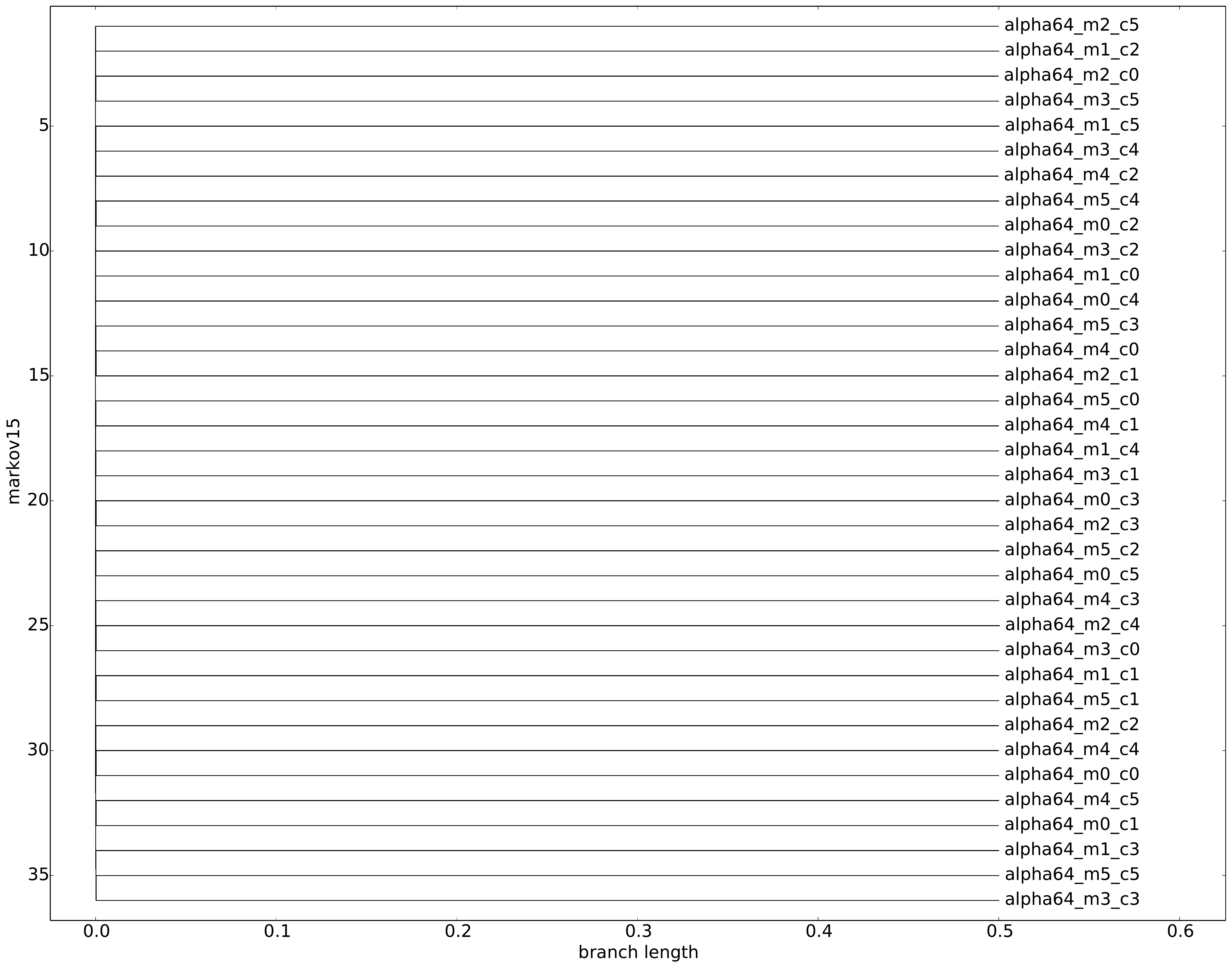}
  \end{minipage}%
  \hfil
  \begin{minipage}[b]{.5\linewidth}
    \includegraphics[width=\textwidth,height=11cm]
        {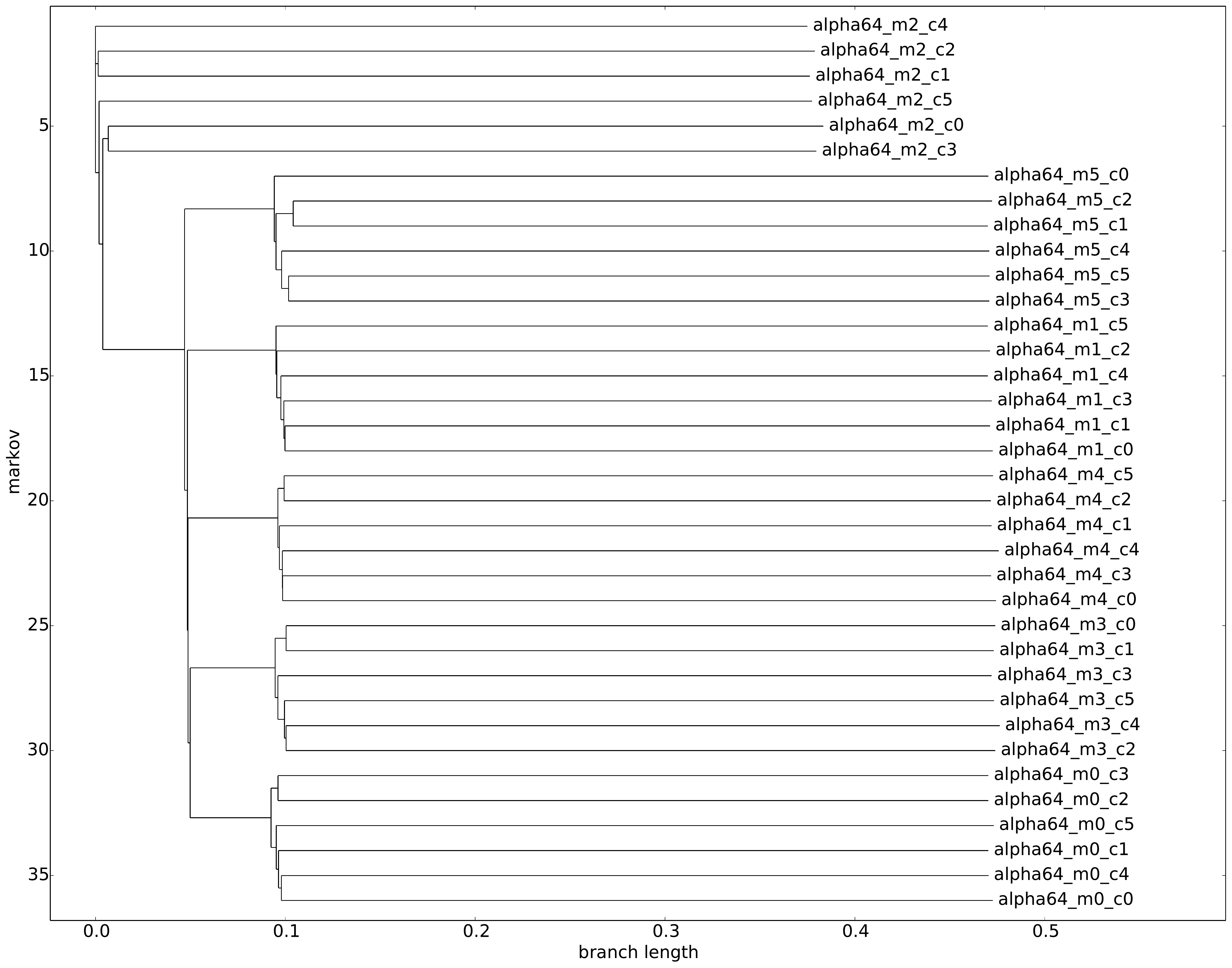}
  \end{minipage}
  \caption{Clustering of Markovian processes using NCD/\texttt{gzip} (left) vs. \NSD{} (right)
    -- $|\alphabet|=64$. The NCD/\texttt{gzip} mostly fails when \salza{} succeeds in correctly
    clustering the Markovian processes with respect to their generative model. }
    \label{fig:markov64}
\end{figure*}

Here, NCD/\texttt{gzip} is unable to correctly classify the realizations with respect to their Markov
transition probabilities matrices -- contrary to \NSD{}, see Fig.~(\ref{fig:markov64}). This further
highlights the stronger discriminative power of \NSD{} over NCD/\texttt{gzip}. The same results have
been obtained using $|\alphabet|=4$ or $|\alphabet|=256$. 

\subsubsection{Larger data set}%\label{ssec:salza-sim-books}

In order to produce meaningful results, the data at \cite{cilibrasi:complearn} is limited to chunks of
16KiB (so that the concatenation of two samples fits the \texttt{gzip} internal buffer length). Because
we are now able to cluster data of arbitrary length, we use to this end a corpus of freely available
French classical books in their entirety, see Fig.~(\ref{fig:corpus}). Most books get classified by
author and the branch lengths of the clusters are intuitively linked to the perception of a native French
reader, thus confirming the cognitive nature of information distances.

Due to the internal, 32KiB limited length of \texttt{gzip} buffer \cite{cebrian:ncd:pitfalls}, the
NCD/\texttt{gzip} fails here as badly as in Fig.~(\ref{fig:markov64}), hence we do not depicts these
results to focus only on \salza{}. The Neighbor-Joining method is known to sometimes produce negative
branch lengths \cite{saitou:neighbor:joining}, this happens here for Montaigne. The unbounded buffer
of \salza{} allows to produce quite a meaningful classification of the books. Note that the tale genre
gets classified as such (Perrault, Aulnoy), although the proper style of Voltaire makes his works a
cluster of their own. \salza{} roughly classifies books by author and by chronological order, thus
reflecting the evolution of the French language itself.

\begin{figure*}%
  \centering
  \begin{minipage}[htb]{.5\linewidth}
    \centering
    \begin{tabular}{|l|l|l|}\hline
      Author & Work & Size \\
      \hline
      Montesquieu      & De l'Esprit des Lois       & 2.1M \\
      Dumas            & Vingt Ans Apr\`es          & 1.7M \\
      Rousseau         & Les Confessions            & 1.6M \\
      Dumas            & Les Trois Mousquetaires    & 1.4M \\
      Montaigne        & Essais, II                 & 1.2M \\
      Stendhal         & La Chartreuse de Parme     & 1.1M \\
      Zola             & Germinal                   & 1.1M \\
      Stendhal         & Le Rouge et le Noir        & 1.1M \\
      Zola             & L'Assomoir                 & 976K \\
      Montaigne        & Essais, III                & 886K \\
      La Bruy\`ere     & Les Caract\`eres           & 847K \\
      Zola             & La B\^ete Humaine          & 787K \\
      Montaigne        & Essais, I                  & 764K \\
      Sade             & Justine                    & 655K \\
      Maupassant       & Bel-Ami                    & 644K \\
      Aulnoy           & Contes, I                  & 634K \\
      La Rochefoucauld & Maximes                    & 616K \\
      Rousseau         & L'\'Emile (I--III)         & 586K \\
      Pascal           & Les Provinciales           & 573K \\
      Montesquieu      & Lettres Persanes           & 541K \\
      Rousseau         & L'\'Emile (IV)             & 486K \\
      Aulnoy           & Contes, 2                  & 457K \\
      Maupassant       & Une Vie                    & 451K \\
      Maupassant       & Fort Comme la Mort         & 439K \\
      Montesquieu      & Lettres Famili\`eres       & 390K \\
      Maupassant       & Mont Oriol                 & 348K \\
      Pascal           & Pens\'ees                  & 341K \\
      Montesquieu      & Consid\'erations           & 337K \\
      Rabelais         & Gargantua                  & 298K \\
      Sade             & Les Infortunes de la Vertu & 298K \\
      Descartes        & \'Elisabeth                & 290K \\
      Voltaire         & Roman et Contes            & 283K \\
      Rabelais         & Pantagruel                 & 273K \\
      Rousseau         & R\^everies Prom. Solitaire & 244K \\
      Descartes        & Les Passions de l'\^Ame    & 214K \\
      Perrault         & Contes                     & 207K \\
      Voltaire         & Candide                    & 194K \\
      Descartes        & Discours de la M\'ethode   & 126K \\
      Zola             & La D\'eb\^acle             & 123K \\
      \hline
    \end{tabular}
    %    \end{table}
  \end{minipage}%
  \hfil
  \begin{minipage}[htb]{.5\linewidth}
    
    \includegraphics[width=\linewidth,height=11cm]
                    {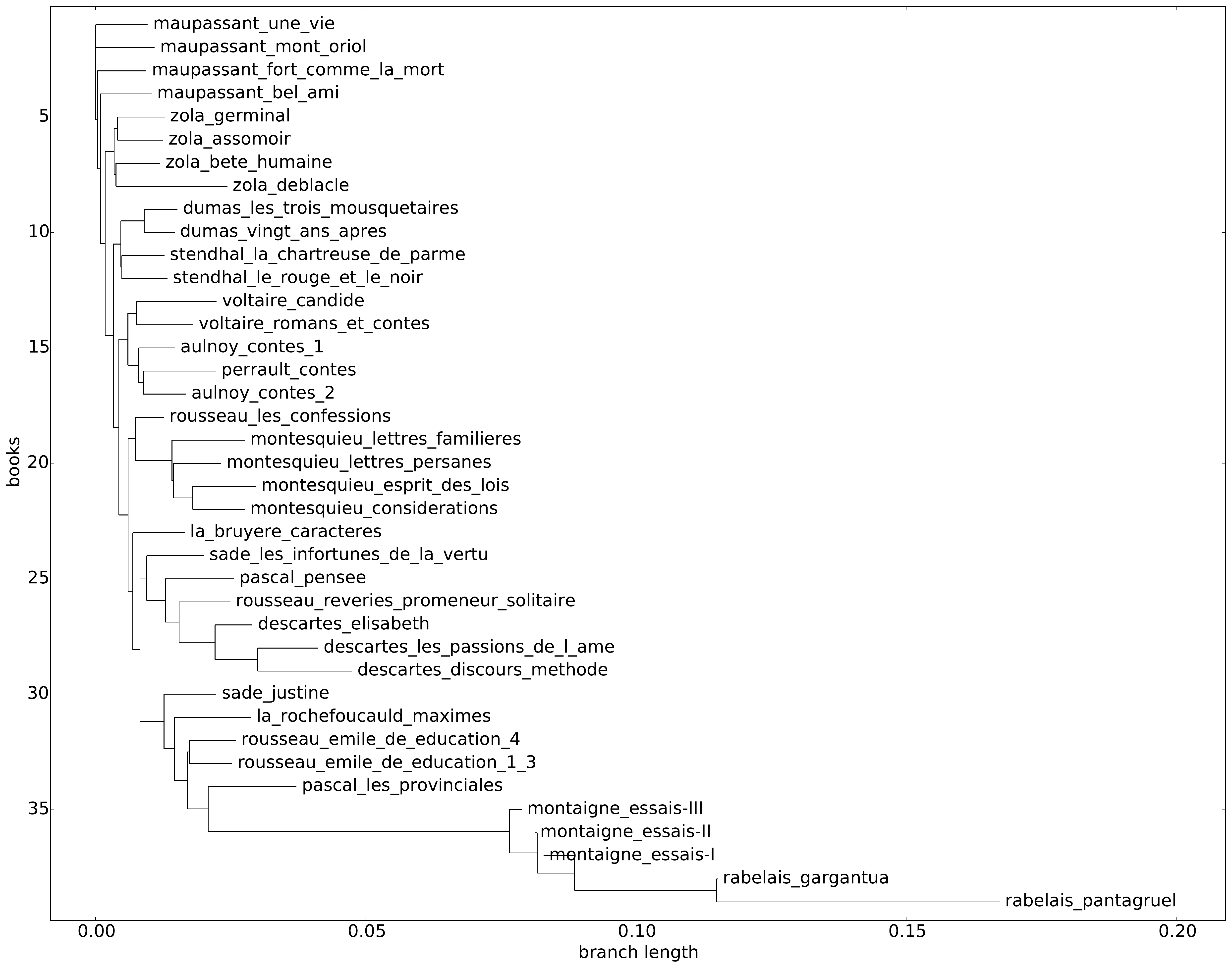}
                    \caption{Clustering of French classic books ranging from XVI\textsuperscript{th} to
                      XIX\textsuperscript{th} century (left) using the \NSD{}.
                    }
                    \label{fig:corpus}
  \end{minipage}
\end{figure*}

\subsection{Causality inference}\label{ssec:app:causality}

Directed information is the tool of choice to perform causality inference. In this section, we
illustrate the performance of the two different kinds of directed information we have defined above
(Sec.~\ref{ssec:dirinfo}). In the first experiment, we concentrate on synthetic data to show how
\salza{} performs in a causal setting: we will recover the structure of several DAGs of random processes
given one of their realizations. This is to simulate the behaviour of systems modeled by dependent time
series. Here, we will naturally use Eq.~(\ref{eq:causal:directed:info}).

In the second experiment, this time on real data, we will use \salza{} to recover the order with which
several versions of a text have been produced by an author. This will illustrate both the usefulness of
full directed algorithmic information in Eq.~(\ref{eq:full:directed:info}) and the accuracy we obtain,
even on a rather small data set. 

\subsubsection{Synthetic data}\label{sssec:app:causality:synthetic}

We have simulated DAGs of random processes using dependent processes: we have specified a connectivity
matrix among processes whose weights are the probabilities of copying a string of length proportional to
the said probability from the designated process. We also have a probability of generating random data.

Let $p_n$ ($0\leq n< N$) be random stationary processes over $\alphabet$. The connectivity matrix
$M\in\mathcal{M}_{N,N+1}\left(\left[0,1\right]\right)$ is such that:
$$\forall 0\leq i< N,\;\;\sum_{j=0}^{N}M_{i,j}=1,$$
the last column of $M$ being the probability of generating random, uniformly distributed symbols. This
means that, for process $p_i$, the probability of copying data from any random point in the past of
process $p_j$ is $M_{i,j}$ ($j< N$). If $j=N$ then we generate random, uniformly distributed
data (the unobserved string that simulates the causal mechanism in the words of
\cite{janzing:causal:inference}). At each step, the length of the data being copied or generated is
proportional to $M_{i,j}$. In order to initiate data generation, we have a short ``burn-in'' period of
12 symbols. The minimum length of data that is to be copied from the past of a process is three symbols
(this is to ensure we can hook to the copied string, see footnote \ref{note:min3bytes}). As
seen from Fig.~(\ref{fig:causality:truth}), we are able to recover the structure of the DAGs. 

\begin{figure*}%
\centering
  \begin{minipage}[b]{.3\linewidth}
    \includegraphics[width=\textwidth]
        {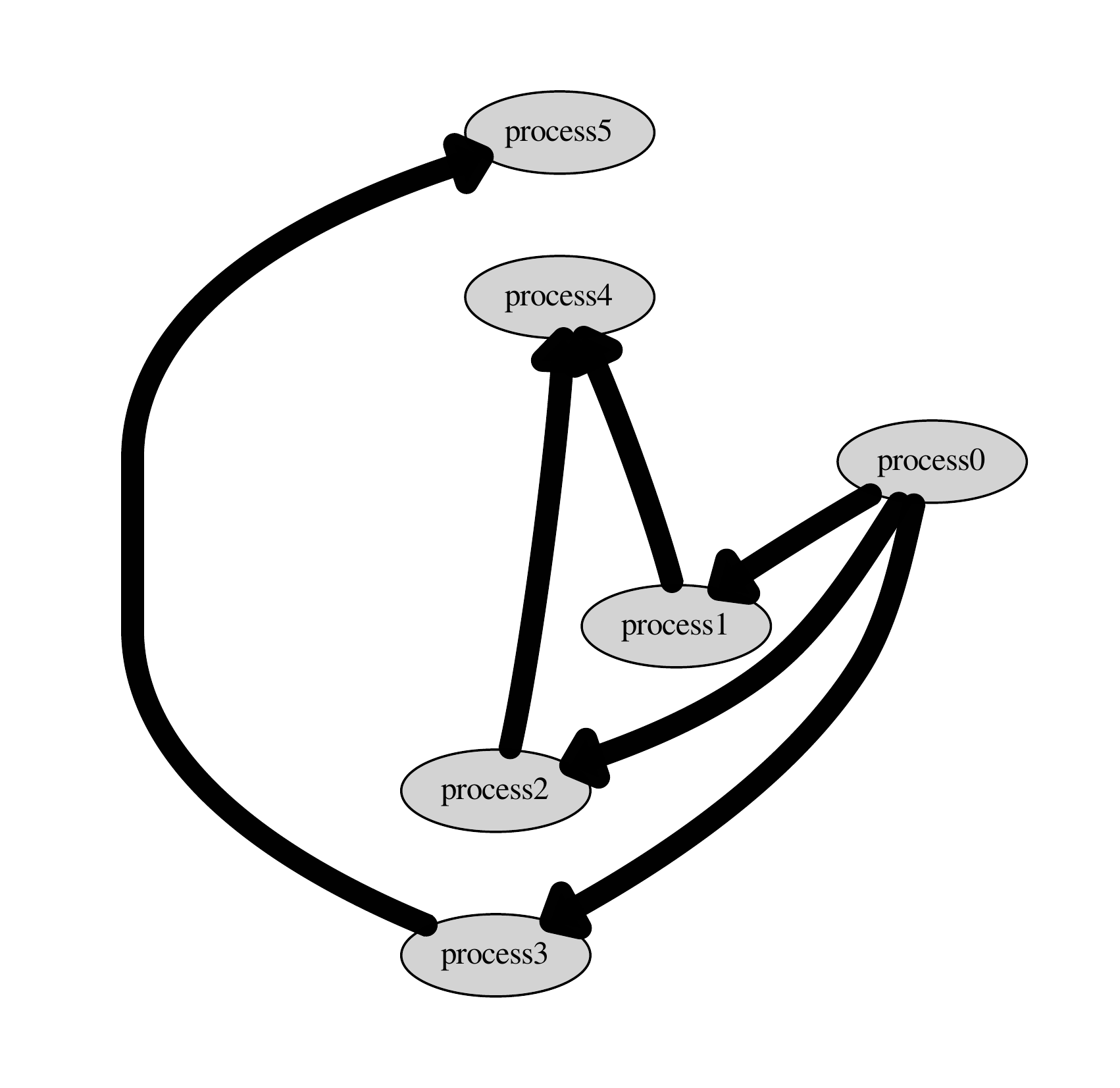}
  \end{minipage}%
  \hfil
  \begin{minipage}[b]{.3\linewidth}
    \includegraphics[width=\textwidth]
        {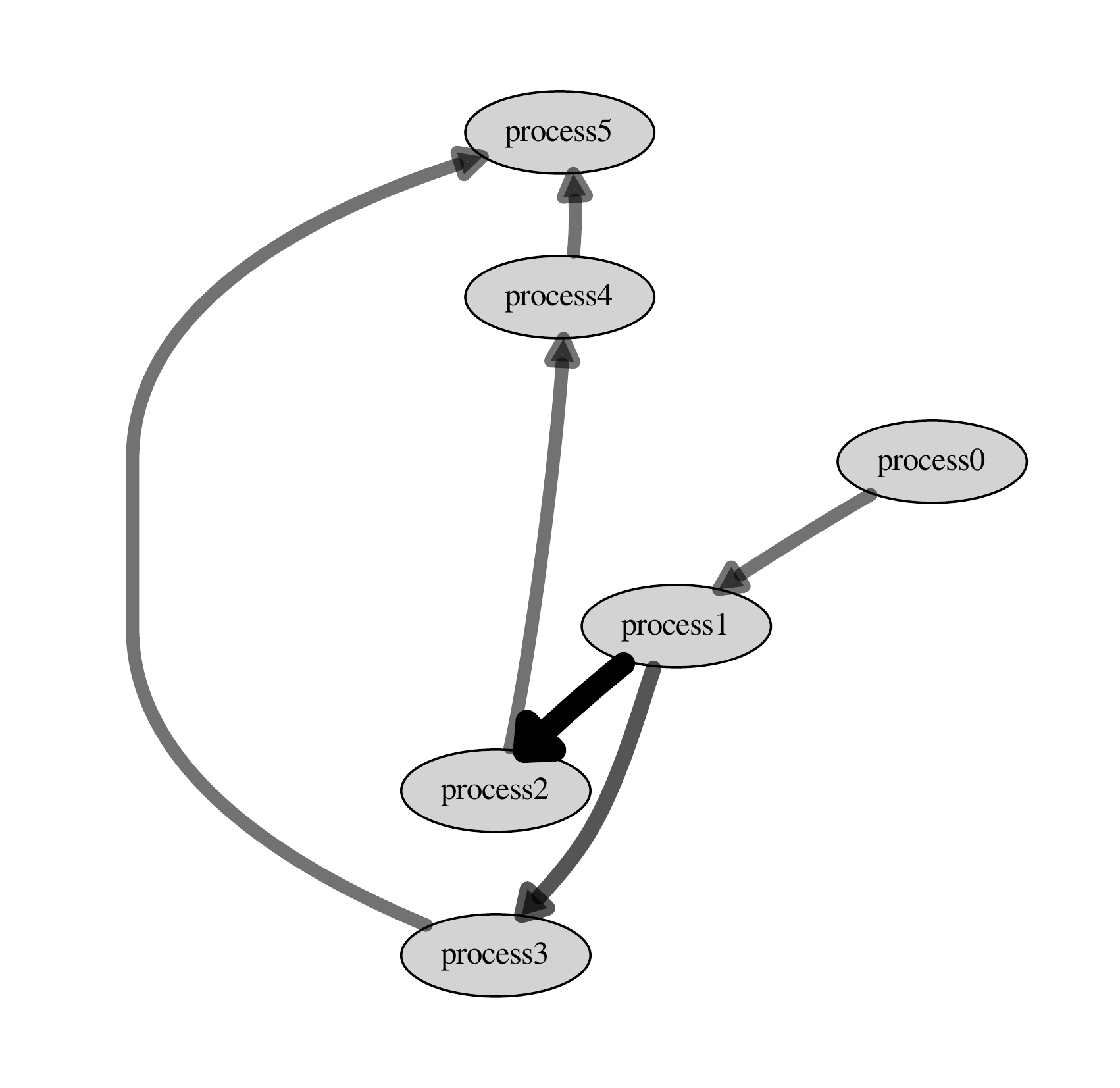}
  \end{minipage}
  \hfil
  \begin{minipage}[b]{.3\linewidth}
    \includegraphics[width=\textwidth]
        {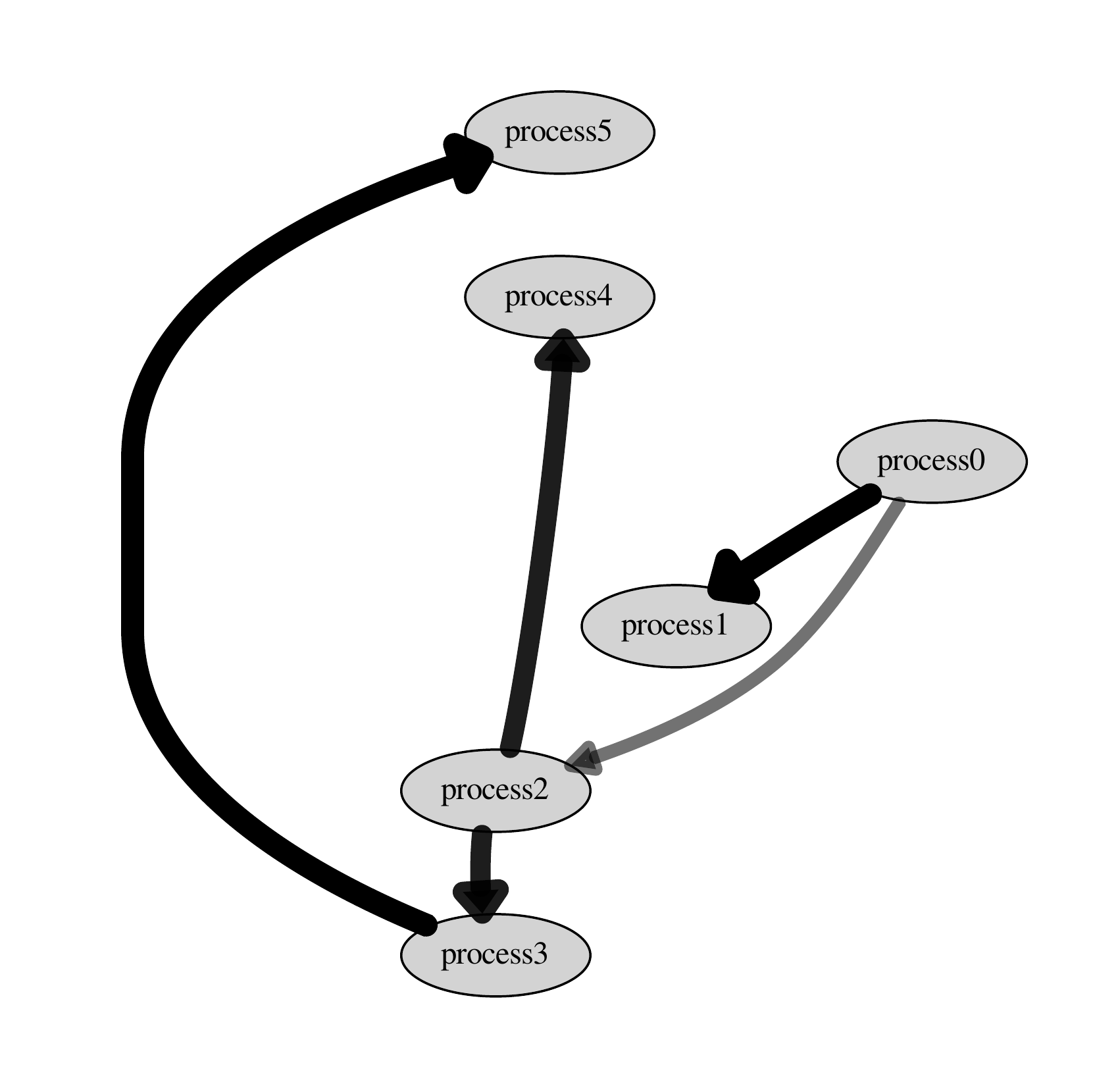}
  \end{minipage} \\
  
  \begin{minipage}[b]{.3\linewidth}
    \tiny
    \begin{tabular}{ll}
      $M=$ & $\left(
      \begin{array}{ccccccc}
        0   & 0   & 0   & 0   & 0 & 0 & 1.0 \\
        .50 & 0   & 0   & 0   & 0 & 0 & .50 \\
        .50 & 0   & 0   & 0   & 0 & 0 & .50 \\
        .50 & 0   & 0   & 0   & 0 & 0 & .50 \\
        0   & .50 & .50 & 0   & 0 & 0 & 0   \\
        0   & 0   & 0   & .50 & 0 & 0 & .50 
      \end{array}\right)
      $ \\
      \hline
      & $\left(
      \begin{array}{cccccc}
        0   & 0   & 0   & 0   & 0 & 0 \\
        .60 & 0   & 0   & 0   & 0 & 0 \\
        .60 & 0   & 0   & 0   & 0 & 0 \\
        .60 & 0   & 0   & 0   & 0 & 0 \\
        0   & .02 & .10 & 0   & 0 & 0 \\
        0   & 0   & 0   & .10 & 0 & 0 
      \end{array}\right)
      $
    \end{tabular}
  \end{minipage}%
  \hfil
  \begin{minipage}[b]{.3\linewidth}
    \tiny
    \begin{tabular}{ll}
      $M=$ & $\left(
      \begin{array}{ccccccc}
        0 0 & 0   & 0   & 0   & 0   & 0 & 1.0 \\
        .50 & 0   & 0   & 0   & 0   & 0 & .50 \\
        0 0 & .90 & 0   & 0   & 0   & 0 & .10 \\
        0 0 & .60 & 0   & 0   & 0   & 0 & .40 \\
        0 0 & 0   & .50 & 0   & 0   & 0 & .50 \\
        0 0 & 0   & 0   & .50 & .50 & 0 & 0 
      \end{array}\right)
      $ \\
      \hline
      & $\left(
      \begin{array}{cccccc}
        0   & 0   & 0   & 0   & 0   & 0 \\
        .60 & 0   & 0   & 0   & 0   & 0 \\
        0   & .10 & 0   & 0   & 0   & 0 \\
        0   & .40 & 0   & 0   & 0   & 0 \\
        0   & 0   & .03 & 0   & 0   & 0 \\
        0   & 0   & 0   & .10 & .04 & 0 
      \end{array}\right)
      $
    \end{tabular}
  \end{minipage}%
  \hfil
  \begin{minipage}[b]{.3\linewidth}
    \tiny
    \begin{tabular}{ll}
      $M=$ & $\left(
      \begin{array}{ccccccc}
        0   & 0 & 0   & 0   & 0 & 0 & 1.0 \\
        .90 & 0 & 0   & 0   & 0 & 0 & .10 \\
        .50 & 0 & 0   & 0   & 0 & 0 & .50 \\
        0   & 0 & .80 & 0   & 0 & 0 & .20 \\
        0   & 0 & .80 & 0   & 0 & 0 & .20 \\
        0   & 0 & 0   & .90 & 0 & 0 & .10 
      \end{array}\right)
      $ \\
      \hline
      & $\left(
      \begin{array}{cccccc}
        0   & 0 & 0   & 0   & 0 & 0 \\
        .50 & 0 & 0   & 0   & 0 & 0 \\
        .40 & 0 & 0   & 0   & 0 & 0 \\
        0   & 0 & .30 & 0   & 0 & 0 \\
        0   & 0 & .30 & 0   & 0 & 0 \\
        0   & 0 & 0   & .10 & 0 & 0 
      \end{array}\right)
      $
    \end{tabular}
  \end{minipage} \\

  \begin{minipage}[b]{.3\linewidth}
    \includegraphics[width=\textwidth]
        {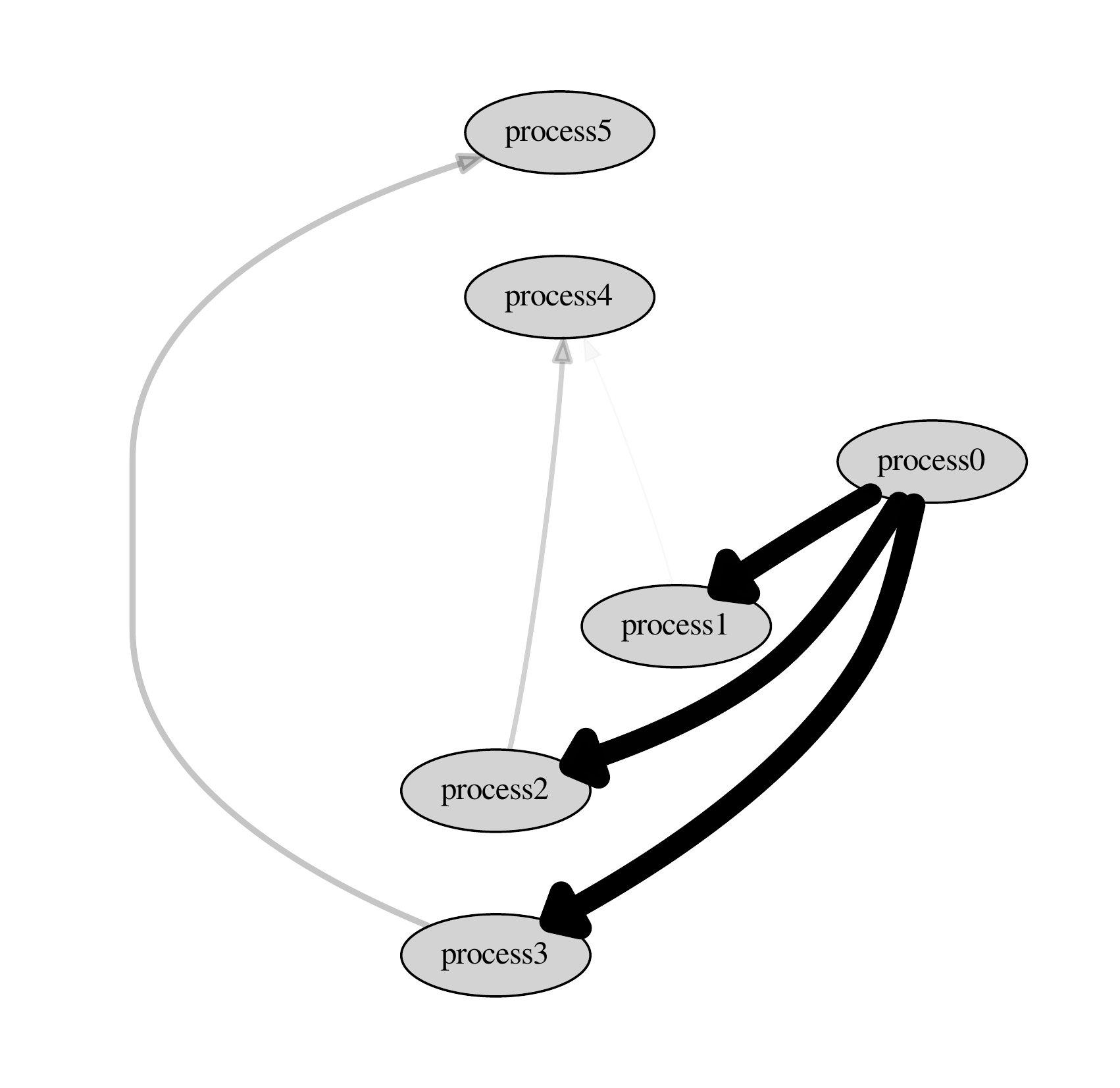}
  \end{minipage}%
  \hfil
  \begin{minipage}[b]{.3\linewidth}
    \includegraphics[width=\textwidth]
        {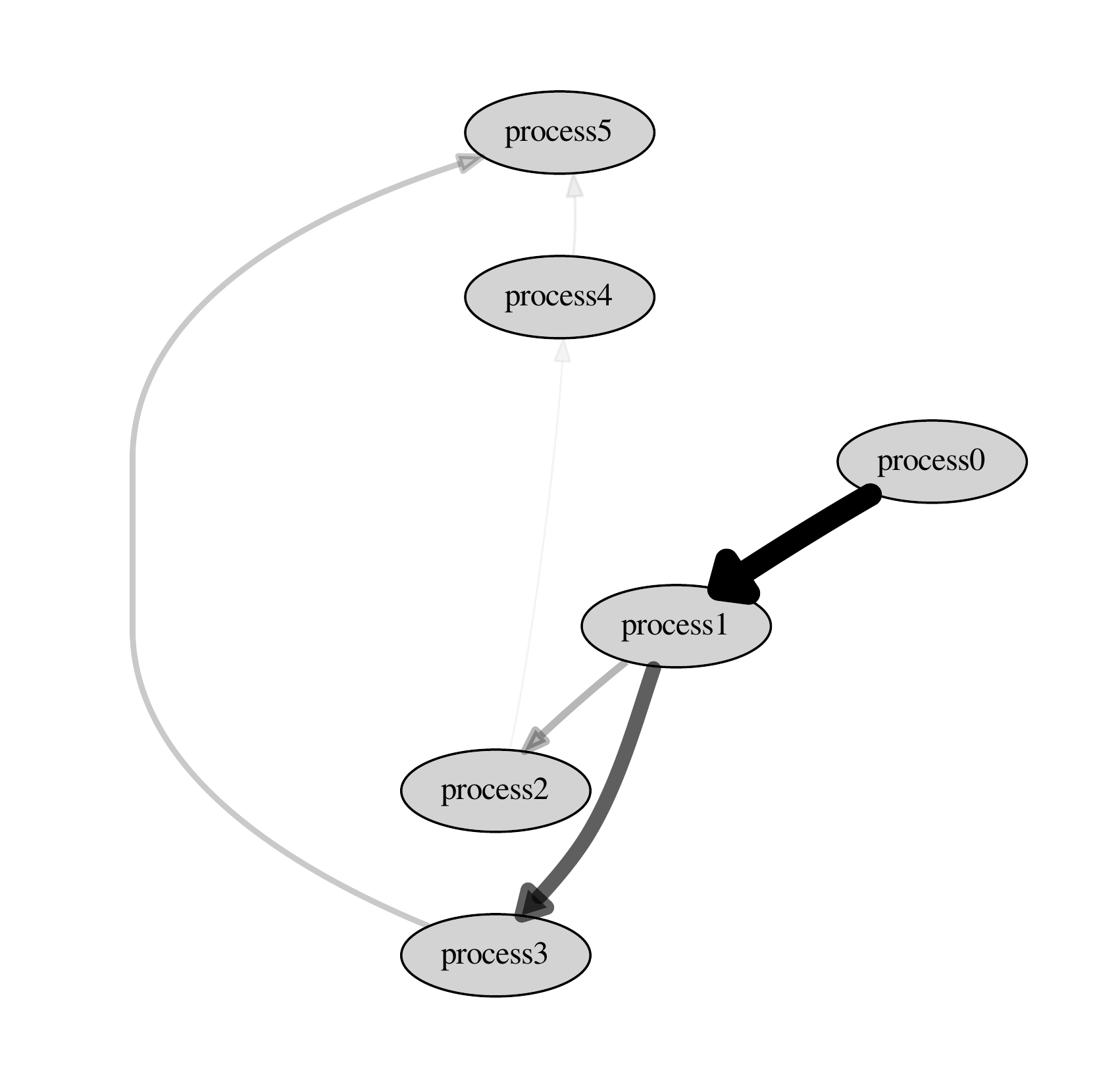}
  \end{minipage}
  \hfil
  \begin{minipage}[b]{.3\linewidth}
    \includegraphics[width=\textwidth]
        {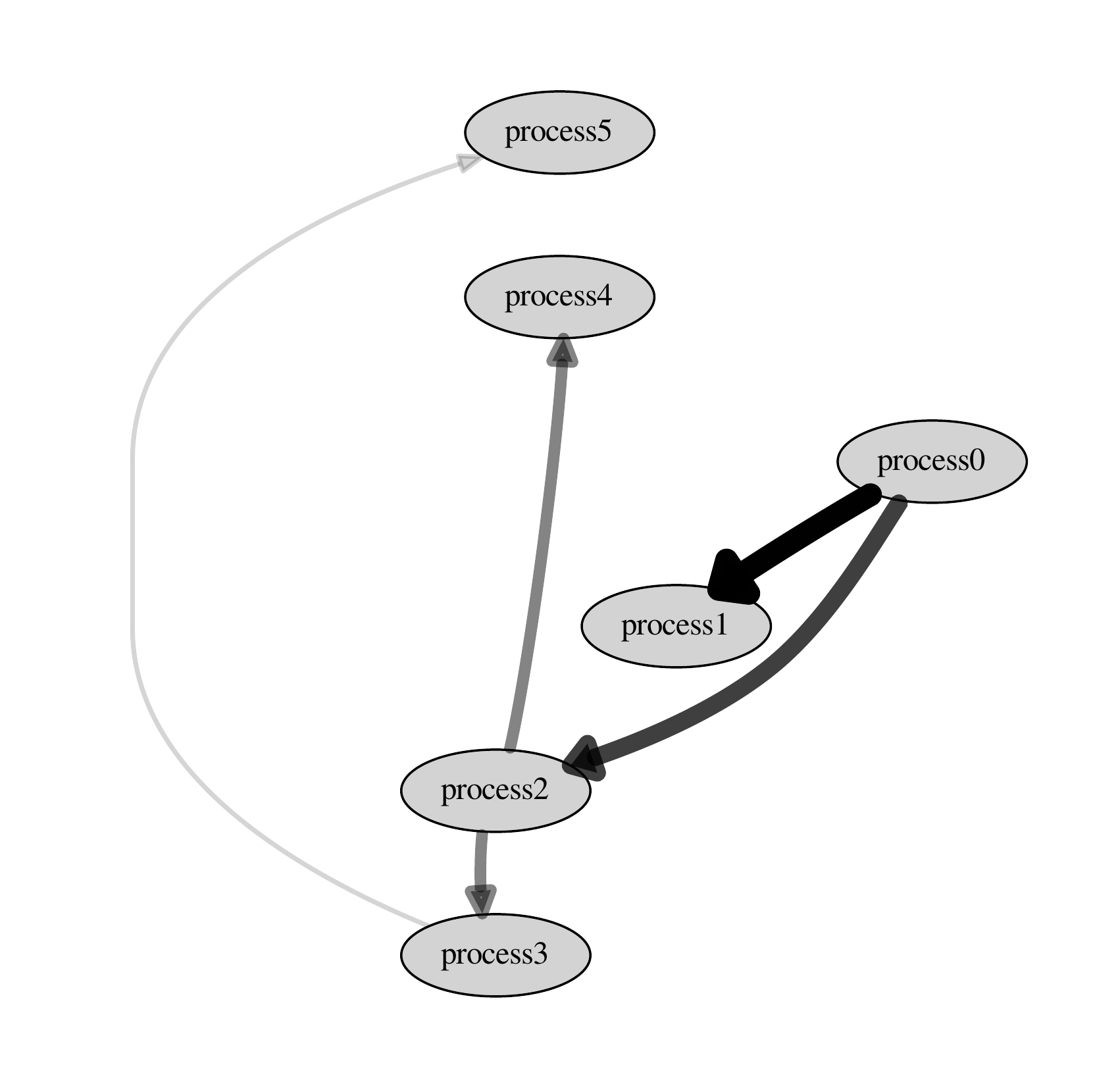}
  \end{minipage} \\
  
  \caption{Top: Three sample target DAGs we want to recover from one of their realizations along with
    their connectivity matrix: the last column is the probability of creating random strings -- this is
    the amount of innovation contained in the strings produced by the said process (and the source of
    randomness used to simulate innovation has not been represented for clarity).
    Bottom: Recovered connectivity matrices using $C_S$ along with their DAGs. We have filtered out
    directed information values such that $C_S(x_i\rightarrow x_j)<5e-3$. 
    Both the thickness and the transparency of the arrows are used to
    depict the amount of innovation created in one process and transmitted to another. This means that
    the structure of the DAG we are recovering is only dictated by the innovation created in every
    single process and transmitted to his neighbor(s). Some arrows may be faint to spot, but the
    recovered structure is correct in all cases. }
  \label{fig:causality:truth}
\end{figure*}

\subsubsection{An experiment in literature}\label{sssec:app:toussaint}

Jean-Philippe Toussaint is a famous Belgian author of French expression with a specific way of writing:
he works by producing paragraphs one after the other. Each paragraph gets typeset, annotated by hand,
typeset again, annotated again, and so on until the author is satisfied. Some of his paragraphs
culminate to more than 50 successive versions. In Fig.~(\ref{fig:reticence:scans}), we show the eight
successive versions of one of his paragraphs (Jean-Philippe Toussaint does not necessarily typeset
exactly the annotated version but makes changes in between). These versions are called fragments in
Fig.~(\ref{fig:reticence:scans}) and Fig.~(\ref{fig:reticence:results}). As one can see in
top and middle plots of Fig.~(\ref{fig:reticence:results}) (clustering using resp. Neighbor-Joining and
UPGMA), our semi-distance allows to correctly recover the chronological order of the fragments. In this
particular application, one would preferably use the UPGMA method because its results are more easily
read into by non specialists. 

On the lower graph of Fig.~(\ref{fig:reticence:results}), all the arrows have been kept in order to allow
in-depth inspection of the amount of differential innovation. This representation is certainly richer
as it allows to grasp the amount of information that has been reused from one fragment to another. 
All three results allow to correctly recover the order with which the fragments have been
actually written by Jean-Philippe Toussaint.

\begin{figure}%
  \centering
  \begin{minipage}[b]{\textwidth}
    \includegraphics[width=.45\textwidth]
                    {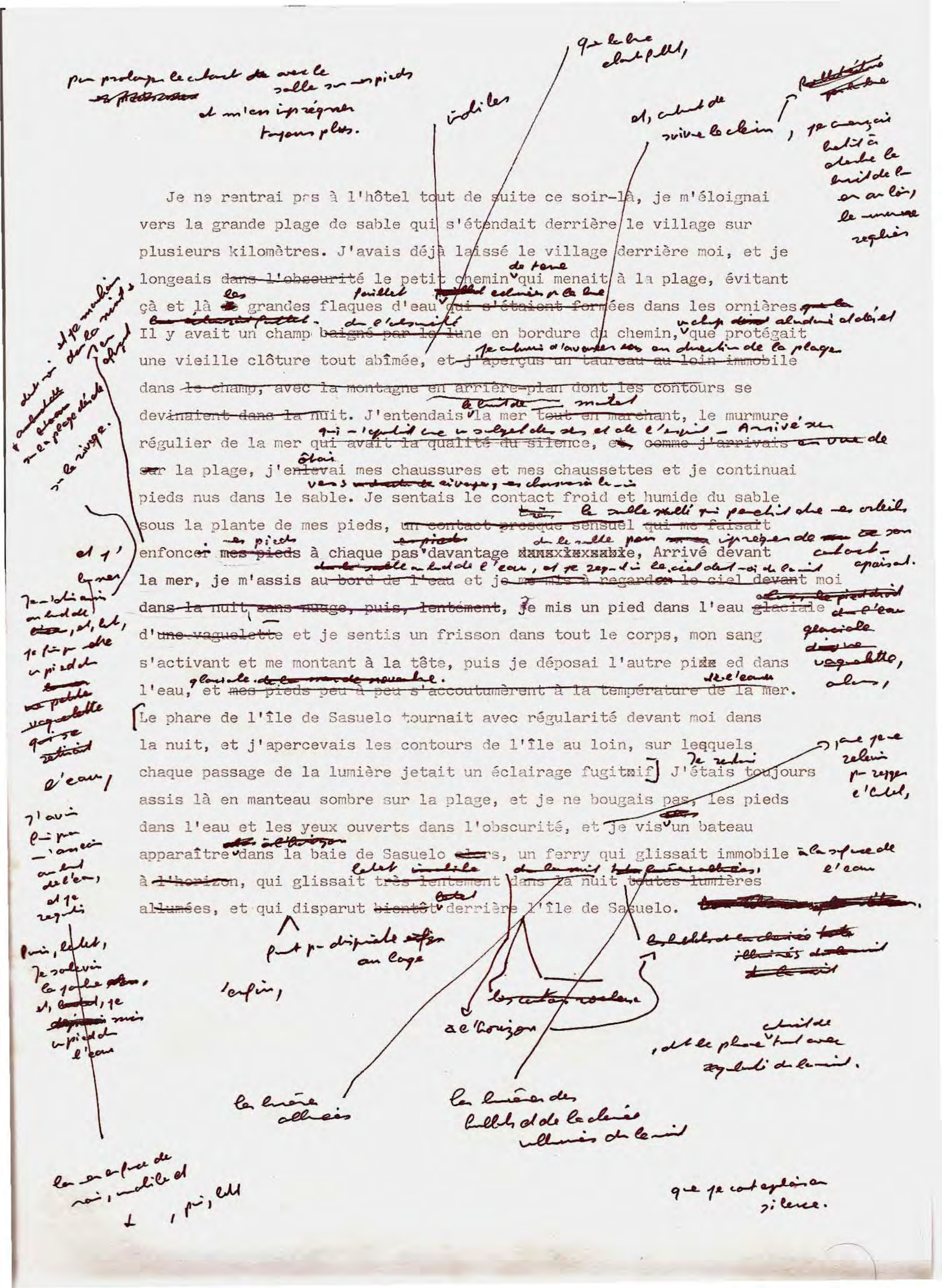}
    \includegraphics[width=.45\textwidth]
                    {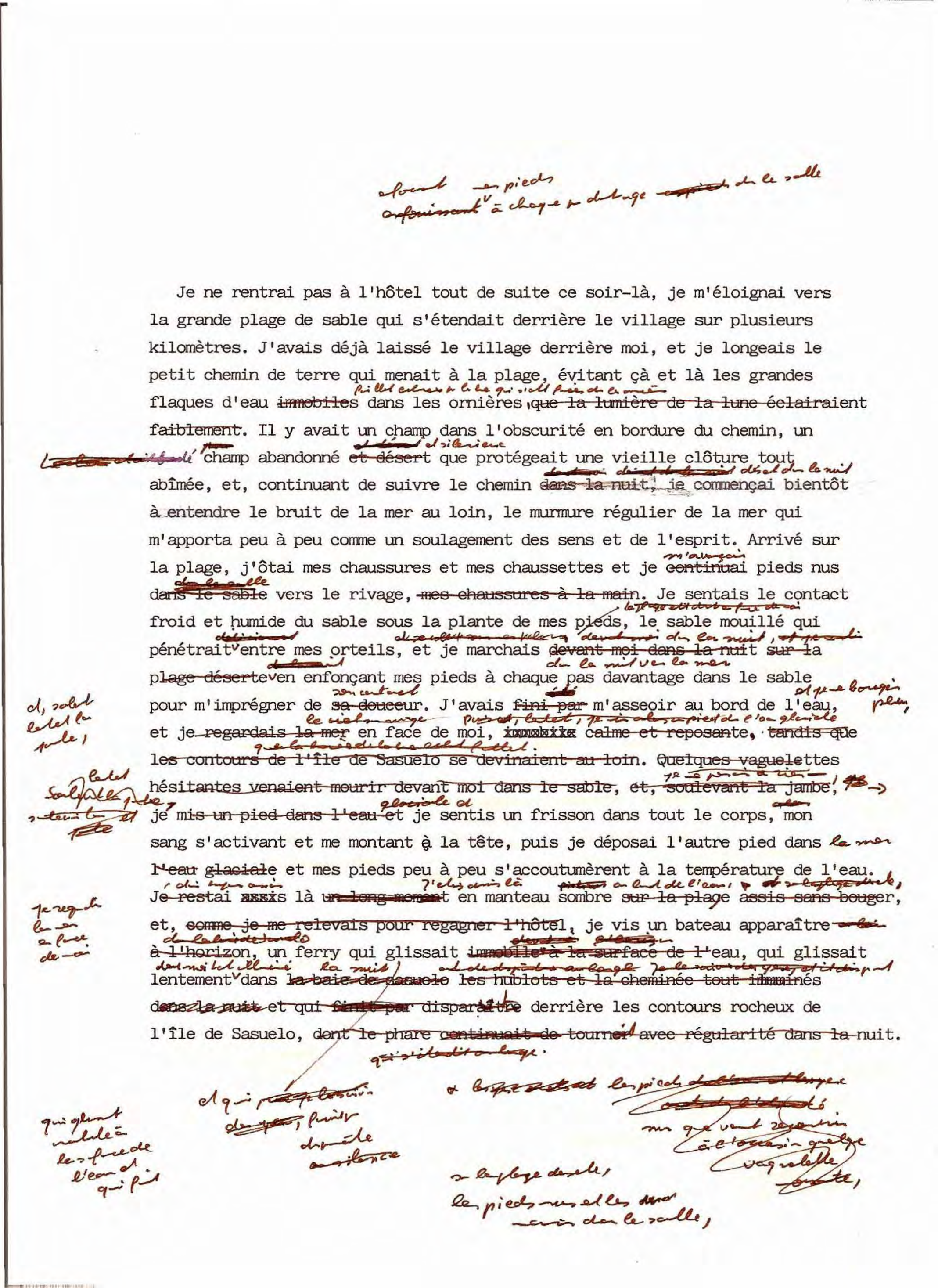}
  \end{minipage}%
  \\
%  \hfil
  \begin{minipage}[b]{\textwidth}
    \includegraphics[width=.45\textwidth]
        {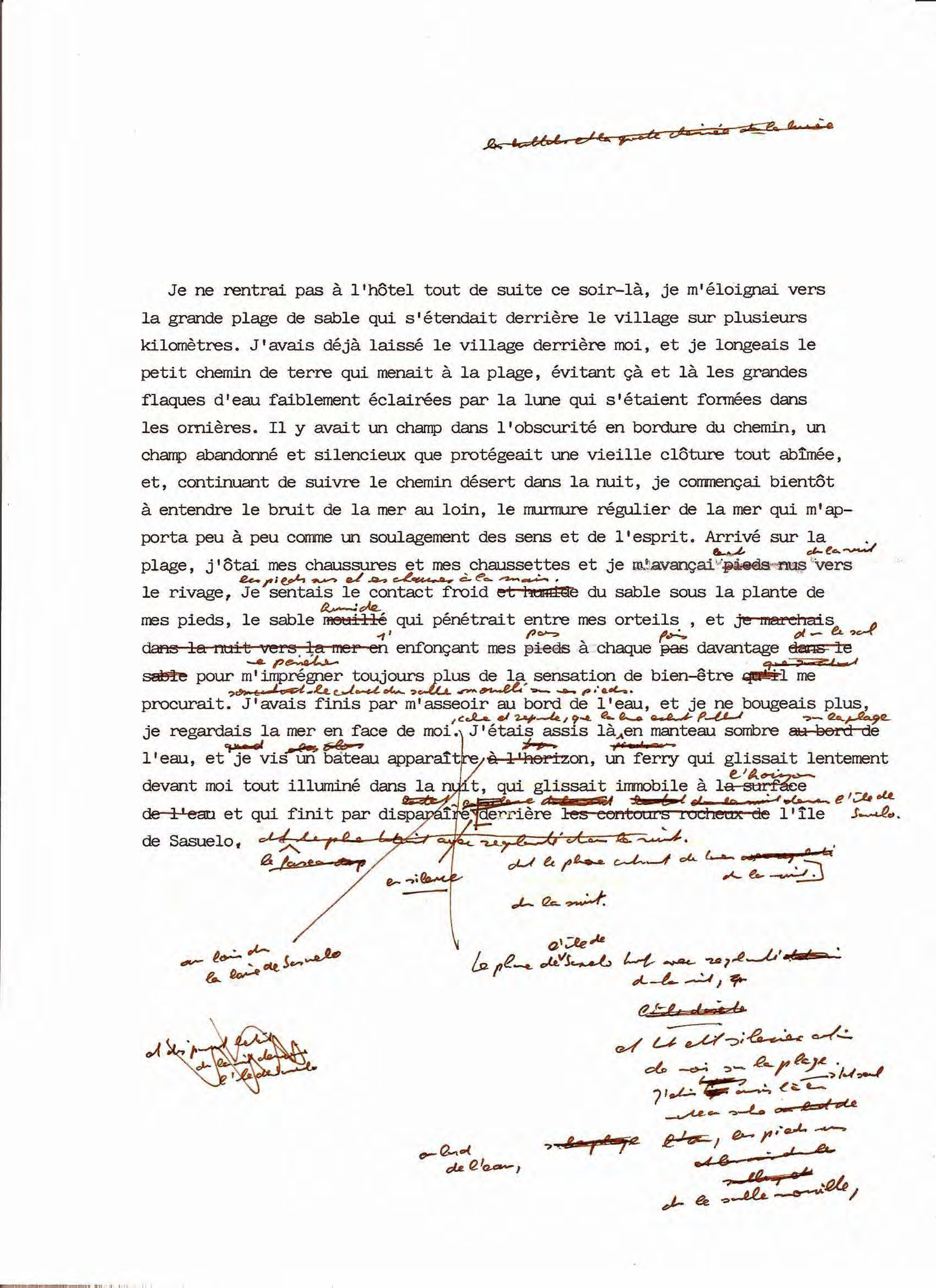}
    \includegraphics[width=.45\textwidth]
        {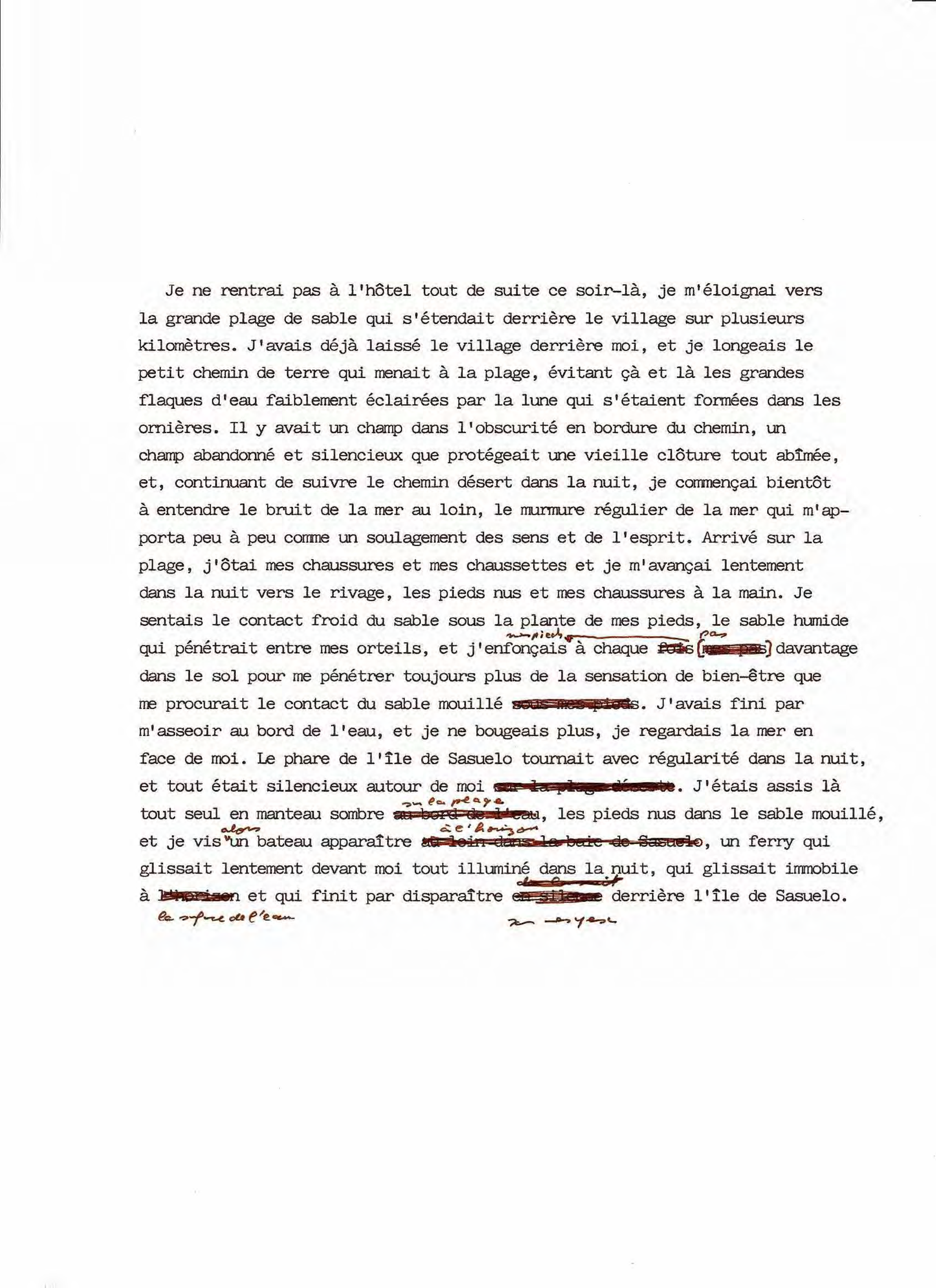}
  \end{minipage}
  
  \caption{Draft scans of the last pages of Jean-Philippe Toussaint's novel {\em La R\'eticence}
    \cite{toussaint:reticence}, by the author (freely available at
    \href{http://jptoussaint.com/documents/9/9b/La\_Réticence,\_brouillons\_dernières\_pages.pdf}
         {\texttt{jptoussaint.com}}). The transcripts we have used are labeled as follows:
         Fragments 1 and 2 are the typeset and annotated versions of the upper-left scan, fragments 
         3 and 4 are the typeset and annotated versions of the upper-right scan,
         fragments 5 and 6 are the typeset and annotated versions of the lower-right scan, and
         fragments 7 and 8 are the typeset and annotated versions of the lower-left scan. }
    \label{fig:reticence:scans}
\end{figure}

\begin{figure}[!ht]
\centering
\includegraphics[width=.5\linewidth]
        {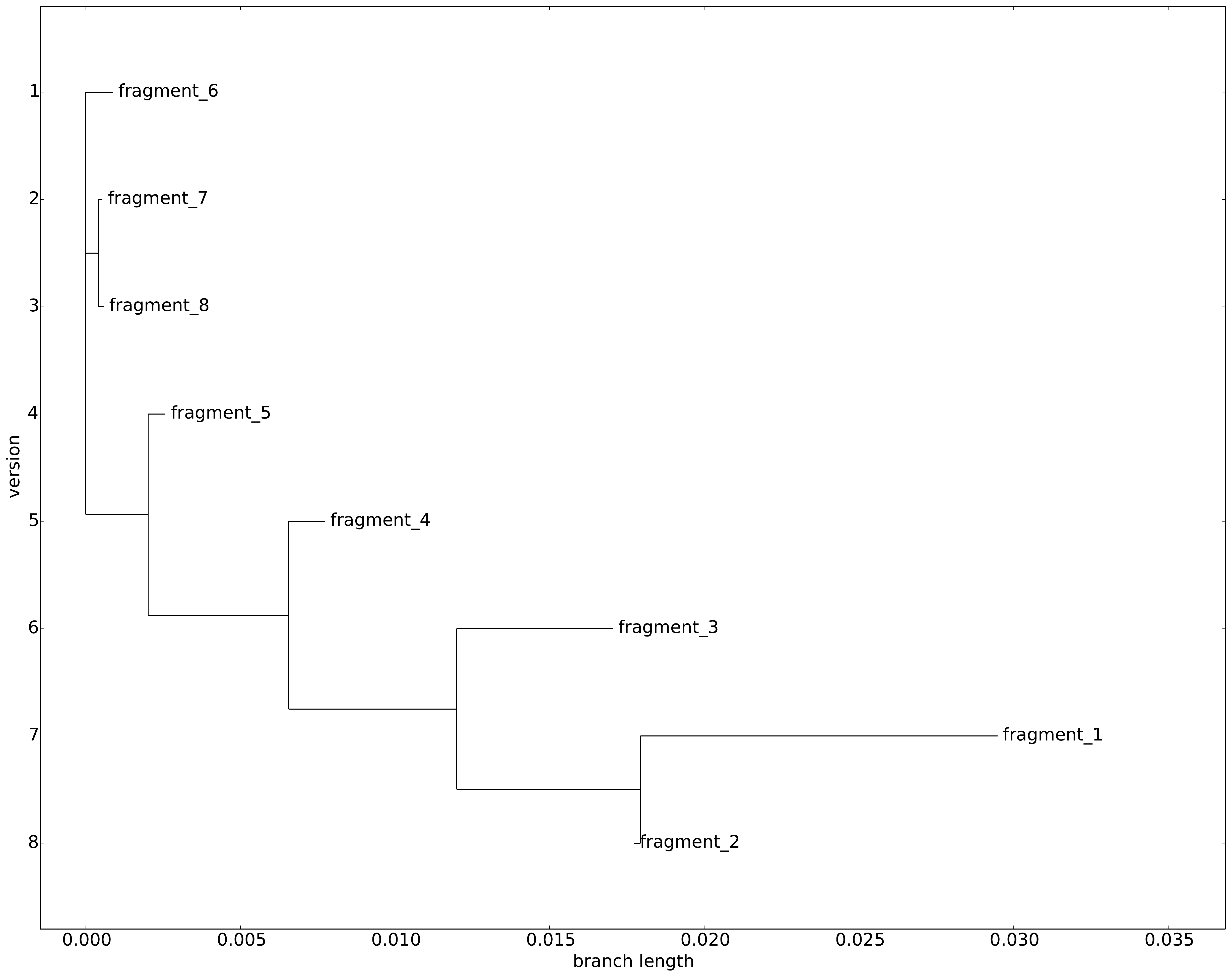} \\
\includegraphics[width=.5\linewidth]
        {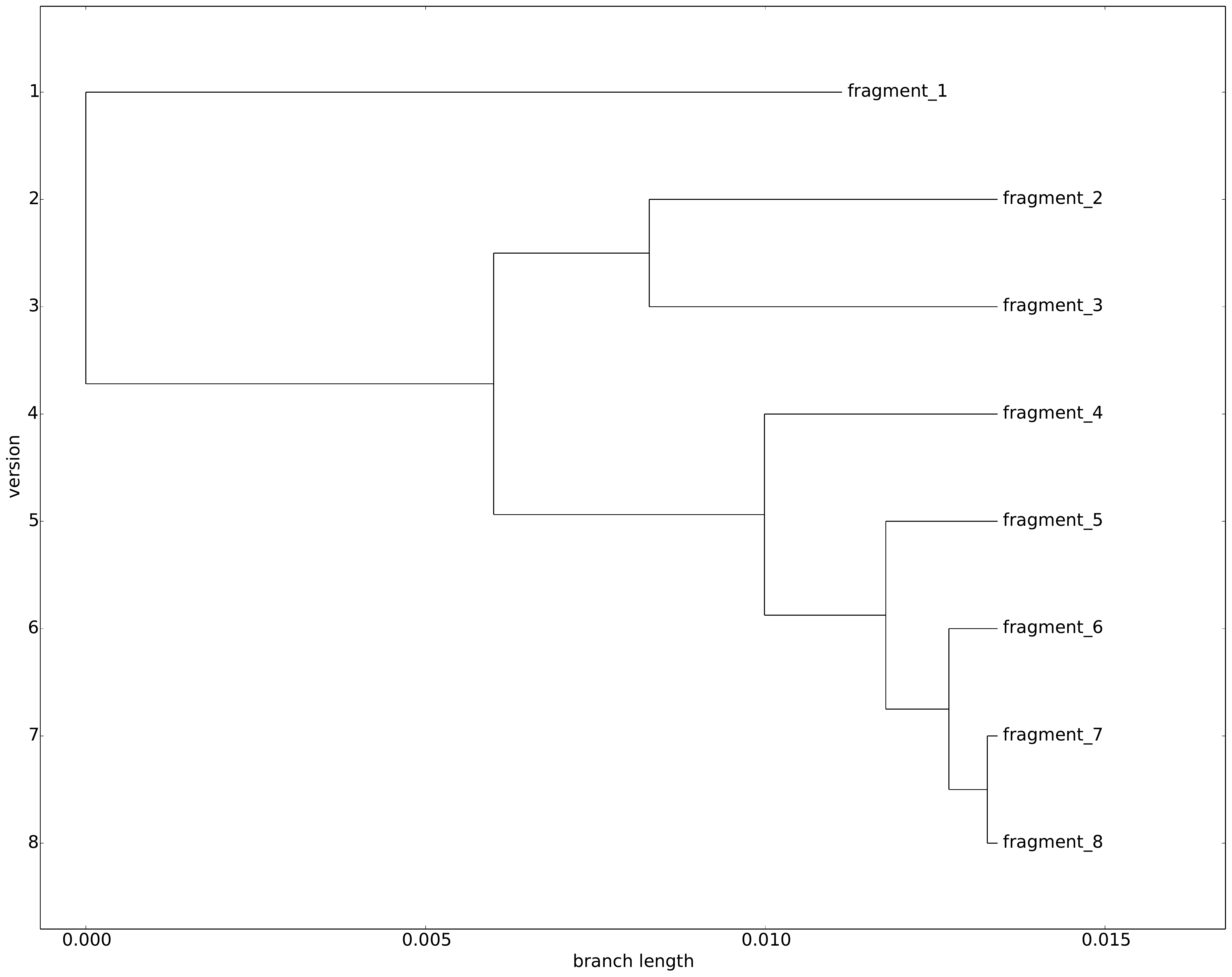} \\
\includegraphics[width=.5\linewidth, height=7cm]
        {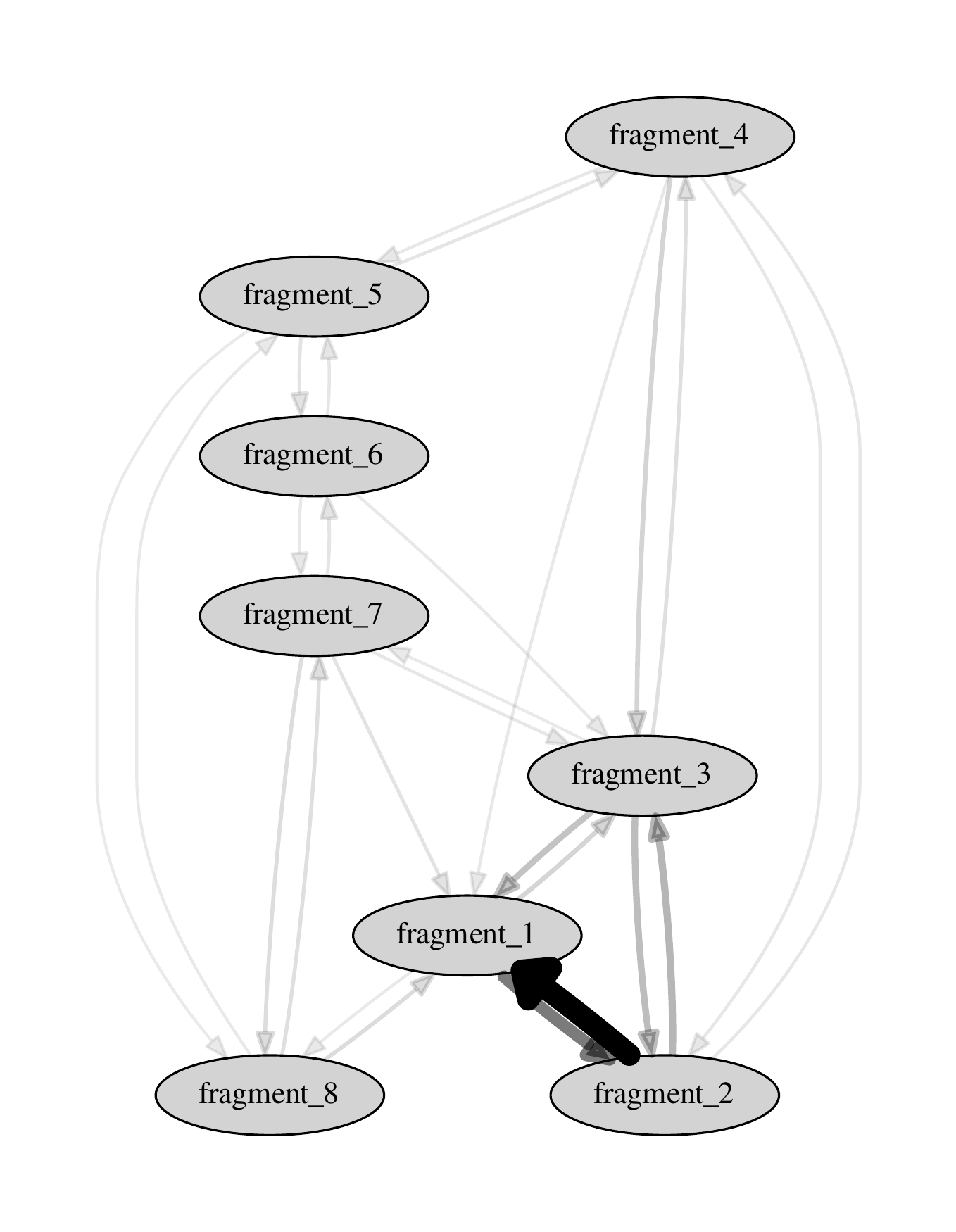}
        \caption{Top: The distance matrix of the data in Fig.~(\ref{fig:reticence:scans}) depicted using
          Neighbor-Joining \cite{saitou:neighbor:joining} clustering (for the sake of completeness). 
          Middle: The distance matrix depicted using UPGMA \cite{sokal:upgma} clustering. This is the
          method of choice one would use in this case. 
          Bottom: The inferred causality graph using $F_S$ (because the author may have -- actually, 
          has -- moved parts of the fragments later in the text).
        }
\label{fig:reticence:results}
\end{figure}

\section*{Acknowledgements}
We are indebted to Profs. Brigitte Combes and Thomas Lebarb\'e (Universit\'e Grenoble Alpes), heads of
the project {\em La R\'eticence} (supported by the \href{http://cahier.hypotheses.org}{CAHIER}
consortium, TGIR \href{http://www.huma-num.fr}{Huma-Num}) for providing the transcripts of
Jean-Philippe Toussaint's drafts. 

\ifCLASSOPTIONcaptionsoff
  \newpage
\fi

\end{document}